\documentclass[11pt]{amsart}
\usepackage[T1]{fontenc}
\usepackage{lmodern,microtype,amsmath,amssymb,amsthm,mathtools,enumitem}
\usepackage[margin=1in]{geometry}
\usepackage{booktabs}
\usepackage[hidelinks]{hyperref}
\newtheorem{mainthm}{Theorem}

\newtheorem{theorem}{Theorem}[section]
\newtheorem{proposition}[theorem]{Proposition}
\newtheorem{lemma}[theorem]{Lemma}
\newtheorem{corollary}[theorem]{Corollary}
\theoremstyle{definition}

\theoremstyle{remark}

\newcommand{\bits}{\{-1,1\}}
\newcommand{\E}{\mathbb E}
\renewcommand{\P}{\mathbb P}

\newcommand{\HH}{\mathsf H}
\newcommand{\I}{\mathsf I}
\newcommand{\D}{\mathsf D}
\DeclareMathOperator{\Var}{Var}
\DeclareMathOperator{\Inf}{Inf}
\DeclareMathOperator{\supp}{supp}
\newcommand{\wh}[1]{\widehat{#1}}
\newcommand{\eps}{\varepsilon}
\newcommand{\symdiff}{\mathbin{\triangle}}
\newcommand{\Lp}{L_*}
\newcommand{\Lmix}{L_{\mathrm{mix}}}
\newcommand{\Clev}{C_{\mathrm{lev}}}
\newcommand{\Creg}{C_{\mathrm{reg}}}
\newcommand{\Cmulti}{C_m}
\setlist{itemsep=3pt,topsep=5pt}
\hypersetup{pdftitle={Fourier entropy of threshold functions and their compositions},pdfauthor={Yitzchak Shmalo}}
\title[Fourier entropy of threshold functions]{Fourier entropy of threshold functions and their compositions}
\author{Yitzchak Shmalo}
\address{Einstein Institute of Mathematics, The Hebrew University of Jerusalem, Jerusalem, Israel}
\email{yitzchak.shmalo@gmail.com}
\date{11 October 2026}
\begin{document}
\begin{abstract}
We prove the Fourier entropy--influence inequality for every linear threshold function, with a constant below $8.937$ and no restrictions on its weights or threshold. The argument applies to regular Boolean functions and gives logarithmic bounds for their Boolean combinations and for unions with bounded variable overlap. Character multiplication preserves the estimates and yields bounds by influence minus variance. Under arbitrary product measures, symmetric thresholds satisfy the strengthened inequality needed for read-once composition, uniformly in depth and fan-in. A joint description of each Fourier set refines the overlap estimate and recovers constant bounds for addressing. Exact selector identities extend the bound to orthogonal multiplexers with shared data. Rare activation preserves the optimal constant of unrestricted FEI. The proofs combine pivotal antichains, spectral reference measures, and product-measure composition.
\end{abstract}
\maketitle
\enlargethispage{3pt}

\section{Introduction}

For a Boolean function $f:\bits^n\to\bits$, write
\[
 f(x)=\sum_{S\subseteq[n]}\wh f(S)\chi_S(x),
 \qquad
 \chi_S(x)=\prod_{i\in S}x_i,
\]
where the underlying measure is uniform unless another measure is specified.
The squared Fourier coefficients form a probability distribution. Its entropy and mean set size are
\begin{equation}\label{eq:HI}
 \HH(f)=\sum_{S\subseteq[n]}\wh f(S)^2\log\frac1{\wh f(S)^2},
 \qquad
 \I(f)=\sum_{S\subseteq[n]}|S|\wh f(S)^2.
\end{equation}
All logarithms denoted by $\log$ are to base $2$; $\ln$ denotes the natural logarithm.
The Fourier entropy--influence conjecture of Friedgut and Kalai \cite{FK96} asserts that
$\HH(f)\le C\I(f)$ for an absolute constant $C$ and every Boolean function.

A linear threshold function is a function
\begin{equation}\label{eq:ltf}
 f(x)=\operatorname{sgn}\left(\sum_{i=1}^n w_i x_i-\theta\right).
\end{equation}
We fix $\operatorname{sgn}(0)=1$; either convention can be represented by moving the threshold slightly off the finite set of attained scores. The weights may vanish, have either sign, or differ by arbitrarily large factors.

Our bounds use the quantity
\begin{equation}\label{eq:D}
 \D(f)=\sum_{i=1}^n\min\{\Inf_i(f),1-\Inf_i(f)\}
      =\min_{U\subseteq[n]}\I(\chi_U f),
\end{equation}
where $\Inf_i(f)=\P[f(X)\ne f(X^{\oplus i})]$.
Multiplication by a coordinate character permutes the Fourier coefficients, so it preserves entropy. In particular, a bound by $\D(f)$ gives FEI for every character multiple of $f$ at once.

A Boolean function is \emph{unate} if it becomes monotone after reversing some input coordinates.
For a monotone function, coordinate $j$ is at least as desirable as coordinate $i$ if
\begin{equation}\label{eq:desire}
 f(x_{j=1,i=-1})\ge f(x_{j=-1,i=1})
\end{equation}
for every assignment to the other coordinates.
We call a unate function \emph{regular} if, in a monotone orientation, its coordinates admit an order in which every earlier coordinate is at least as desirable as every later coordinate.
Thus every linear threshold function is regular: orient the weights to be nonnegative and order them by decreasing magnitude.
Here regularity refers to a coordinate order, and imposes no smallness condition on the largest weight.

Put
\begin{equation}\label{eq:constants}
 \begin{gathered}
 \Clev=\log(4+e),\qquad \Lp=1+\ln2,\qquad \Lmix=2+\gamma,\\
 \Creg=\Clev+2\log(e\pi)<8.937.
 \end{gathered}
\end{equation}
Here $\gamma$ is Euler's constant.

\begin{mainthm}\label{thm:A}
Every regular Boolean function $f$ satisfies
\begin{equation}\label{eq:mainreg}
 \HH(f)\le\Creg\,\D(f).
\end{equation}
Consequently, for every $U\subseteq[n]$, the function $g=\chi_Uf$ satisfies
\[
 \HH(g)\le\Creg\,\I(g),
 \qquad
 \HH(g)\le4\Creg\bigl(\I(g)-\Var(g)\bigr).
\]
These conclusions apply, in particular, to every linear threshold function.
\end{mainthm}

The same method allows several coordinate orders to coexist. The component functions may depend on all the same variables, and their orientations may differ.

\begin{mainthm}\label{thm:B}
Let $m\ge1$, let $h_1,\ldots,h_m$ be regular Boolean functions, let $F:\bits^m\to\bits$ be arbitrary, and put $f=F(h_1,\ldots,h_m)$. Then
\begin{align}
 \HH(f)&\le \left(11.238+\frac{17}{5}\log m\right)\I(f),\label{eq:multiI}\\
 \HH(f)&\le \left(11.460+\frac{22}{5}\log m\right)\D(f).\label{eq:multiD}
\end{align}
For every character multiple $g=\chi_U f$, the second estimate also gives
\[
 \HH(g)\le \left(11.460+\frac{22}{5}\log m\right)\I(g)
\]
and
\[
 \HH(g)\le4\left(11.460+\frac{22}{5}\log m\right)
                  \bigl(\I(g)-\Var(g)\bigr).
\]
The decimal constants in these statements are rounded upward.
\end{mainthm}

This includes intersections, unions, parity combinations, and arbitrary Boolean combinations of $m$ halfspaces.
It also includes every Boolean labeling of a linear score with at most $m$ transitions.
The number of input variables and the geometry of the weights do not enter the bounds.

The proof extends further to a unate function whose coordinates split into $w$ desirability chains: one adds $\log w$ to the constant in Theorem~\ref{thm:A}. For combinations of $m$ such functions, one adds $\log w$ to the constants in Theorem~\ref{thm:B}; see Theorem~\ref{thm:width}.
A different spectral reference measure gives a bound depending on variable overlap rather than the number of component functions.

\begin{mainthm}\label{thm:readk-intro}
Suppose $f$ is a union or intersection of functions $h_1,\ldots,h_N$, each an arbitrary Boolean combination of at most $M$ regular functions, and every input variable is relevant to at most $k$ of the $h_l$. Then
\[
 \HH(f)\le\left(41+7\log k+\frac{29}{5}\log M\right)\I(f).
\]
Writing $\HH^+(f)=\sum_{S\ne\varnothing}\wh f(S)^2\log(1/\wh f(S)^2)$, one also has
\[
 \HH^+(f)\le
 \left(132+\frac{118}{5}\log k+\frac{106}{5}\log M\right)
 \bigl(\I(f)-\Var(f)\bigr).
\]
The number $N$ and the arities of all gates are unrestricted.
\end{mainthm}

The case $M=1$ gives a bound depending only logarithmically on the variable overlap for unions and intersections of any number of halfspaces.
The proof uses a reference distribution generated by a queue of gates.
For each Fourier set, a directed spanning forest supplies one finite run; the disjoint intersections of its root supports with that Fourier set control the cost of activating rare gates.
Signed pivotal pieces provide the same rank estimate as for one regular function, while a fixed-noise hypercontractive bound forces a rare gate's spectral sample to include enough of its variables.

For a product measure with probabilities $q_i=\P(y_i=1)$ on $\{0,1\}^n$, write $\wh f_\mu$ for the coefficients in the normalized product basis and put
\[
 \HH_\mu^+(f)=\sum_{S\ne\varnothing}\wh f_\mu(S)^2
    \log\frac{\prod_{i\in S}4q_i(1-q_i)}{\wh f_\mu(S)^2},
 \qquad
 \Delta_\mu(f)=\I_\mu(f)-\Var_\mu(f).
\]
The normalization and the treatment of deterministic inputs are given in Section~\ref{bias:normalization}. On the uniform cube, $\HH^+$ is the contribution of the nonempty Fourier coefficients to $\HH$. We say that $f$ satisfies FEI+ with constant $C$ under $\mu$ if $\HH_\mu^+(f)\le C\Delta_\mu(f)$; on the uniform cube this reads $\HH^+(f)\le C(\I(f)-\Var(f))$.

\begin{mainthm}\label{thm:formulas-intro}
Every symmetric threshold satisfies
\[
 \HH_\mu^+(f)\le C_T\Delta_\mu(f),
 \qquad C_T=80+\frac{460}{\ln2}<744,
\]
under every product measure. The same inequality holds for every read-once formula over symmetric threshold and parity gates, with arbitrary depth, fan-in, and negations. For uniform original inputs these formulas satisfy
\[
 \HH(f)\le744\I(f),\qquad
 \HH(f)\le747\bigl(\I(f)-\Var(f)\bigr).
\]
\end{mainthm}

Here a symmetric threshold is a gate $\mathbf1[\sum_i y_i\ge t]$, with its output identified with $\{-1,1\}$. Read-once means that each original variable occurs in at most one leaf. The gate estimate is uniform over arities and induced output biases. The proof and further formula consequences appear in Sections~\ref{bias:section} and~\ref{comp:section}.

The joint witness estimate in Proposition~\ref{prop:joint-witness}
retains the cost of the gates used by an entire Fourier set and gives a
constant bound for addressing inside the same queue construction.
Section~\ref{sec:selectors} derives exact selector identities and
constant-preserving closure under suitable Fourier-support separation.
It also shows that unrestricted two-cell selectors and arbitrarily
rare activations preserve the optimal constant of full FEI. These
reductions specify the scope of the structured hypotheses.

\subsection{Earlier work}

O'Donnell, Wright and Zhou \cite{OWZ11} proved FEI for symmetric and fixed-part symmetric functions and for read-once decision trees. Their method relates derivative noise stability to spectral entropy; the spectral-level argument below refines their bound. Chakraborty, Kulkarni, Lokam and Saurabh \cite{CKLS16} proved, among other results, an $O(\sqrt n)$ entropy bound for linear threshold functions. The 2018 revision of their ECCC report removes an earlier claim of FEI for all linear threshold functions. Chakraborty, Karmalkar, Kundu, Lokam and Saurabh \cite{CKKLS18} proved FEI with high probability for specified random models of linear threshold functions.

Shalev \cite{Shalev18} studied the read-$k$ DNF question and proved a min-entropy result for a restricted class of such formulas. Arunachalam, Chakraborty, Kouck\'y, Saurabh and de Wolf \cite{ACKSW21} established the min-entropy conjecture for read-$k$ DNF with fixed $k$. Theorem~\ref{thm:readk-intro} bounds the full Shannon entropy and applies more generally to unions and intersections of regular functions. Fourier concentration for small-read DNF was studied by Lecomte and Tan \cite{LT21}; Section~\ref{sec:consequences} records a direct concentration consequence of the entropy bound.

O'Donnell and Tan \cite{ODT13} introduced the product-measure strengthening of FEI used in composition and proved its closure under disjoint substitution. They obtained FEI for read-once formulas over any fixed finite gate basis. Our formula results use their composition principle; the additional input is a bound uniform in the arity of symmetric threshold gates, together with the threshold and regular-gate bounds above.

Keller, Mossel and Schlank \cite{KMS12} obtained general entropy bounds involving the coordinate influences; Han \cite{Han25} gives another proof of a bound of this form. Recent results also include the restriction methods of Gonz\'alez, MacManus and Pereyra \cite{GMP26} and Potapov's sharp bound for plateaued functions \cite{Potapov26}. The argument below instead controls a global reference measure on Fourier sets through pivotal antichains.

A complementary geometric case is the author's work on planar percolation \cite{Shmalo26}. It establishes entropy--influence comparison for crossing functions using spatial refinement and arm estimates, and also controls the entropy of pivotal sets. The spectral crossing question was stated by Garban, Pete and Schramm \cite[Section~11]{GPS10}; Kalai \cite[Conjecture~8]{Kalai18} asked for both spectral and pivotal entropy bounds.

\subsection{The argument}

The pivotal assignments of a regular function form an antichain on the coordinates preceding the pivotal coordinate. Resampling part of that predecessor set therefore gives a noise estimate relative to the mass of the derivative. In Fourier space, this bounds the expected logarithmic discrepancy between a coordinate's rank and the number of predecessors present in a spectral set.

A reference distribution on sets of a fixed size converts that rank bound into entropy. Several regular functions are handled by partitioning each pivotal set among its component functions and using the least expensive of their coordinate orders. To obtain the bound by $\D$, we multiply by the character that minimizes influence. This changes signs, but preserves the antichain supports of all derivative pieces outside the selected character. The remaining coordinates receive zero rank cost.

\section{Fourier and entropy preliminaries}\label{sec:prelim}

We record the identities used throughout. Standard background is given in \cite{OD14}.
For a function on the uniform cube, set
\[
 d_i f(y)=\frac{f(y,1)-f(y,-1)}2.
\]
For Boolean $f$, the derivative takes values in $\{0,\pm1\}$ and
\begin{equation}\label{eq:derivative}
 \|d_i f\|_2^2=\Inf_i(f),\qquad
 \wh{d_i f}(T)=\wh f(T\cup\{i\}),\qquad
 \Var(f)=\sum_{S\ne\varnothing}\wh f(S)^2.
\end{equation}
Whenever $\alpha_i=\Inf_i(f)>0$, the probability law
\begin{equation}\label{eq:derivlaw}
 \nu_i(T)=\frac{\wh{d_i f}(T)^2}{\alpha_i}
\end{equation}
is the law of the Fourier set with $i$ deleted, conditional on containing $i$.

For a set of coordinates $U$, let $Y$ be obtained from a uniform $X$ by resampling each coordinate of $U$ with probability $t$, independently. Parseval's identity gives
\begin{equation}\label{eq:partialnoise}
 \E[d(X)d(Y)]=\sum_T\wh d(T)^2(1-t)^{|T\cap U|}.
\end{equation}
The right side is nonnegative even when $d$ takes both signs.

We shall use Gibbs' inequality also for subprobability reference measures. If $P$ is a probability law and $Q$ is nonnegative, $\sum Q\le1$, and $Q>0$ on the support of $P$, then
\begin{equation}\label{eq:gibbs}
 H(P)\le \E_P\log\frac1{Q}.
\end{equation}
This follows by normalizing $Q$ and applying nonnegativity of relative entropy.

\begin{lemma}[Spectral-level entropy]\label{lem:levelentropy}
If $\mathcal S$ has law $\P(\mathcal S=S)=\wh f(S)^2$, then
\[
 H(|\mathcal S|)\le\Clev\I(f),\qquad \Clev=\log(4+e).
\]
\end{lemma}
\begin{proof}
Put $K=|\mathcal S|$ and $v=\Var(f)=\P(K>0)$.
If $v=0$, the conclusion is immediate.
As in the spectral-level argument of \cite[Theorem~5]{OWZ11},
the edge-isoperimetric inequality gives
$\I(f)\ge2p\log(1/p)$, where $p=\min\{\P(f=1),\P(f=-1)\}$.
It implies
\[
 \I(f)\ge v+\tfrac12v\log(1/v).
\]
For completeness, after writing $q=1-2p=\sqrt{1-v}$, the latter comparison is equivalent to
\[
 (q-1)\log(1-q)+(1+q)\log(1+q)-2q\ge0.
\]
The left side vanishes at $0$ and $1$ and has second derivative
$-2q/((1-q^2)\ln2)\le0$.
Since $-(1-v)\ln(1-v)\le v$,
\[
 h_2(v)\le v\log(1/v)+(\log e)v
          \le2\I(f)+(\log e-2)v.
\]
For $\rho=4/(4+e)$, use the geometric reference
$q_k=(1-\rho)\rho^{k-1}$ on the positive integers. Gibbs' inequality gives
\[
 vH(K\mid K>0)
 \le(\I(f)-v)\log\frac1\rho+v\log\frac1{1-\rho}.
\]
The coefficient of $v$ in the sum of these bounds is
$\log e-2+\log(\rho/(1-\rho))=0$.
The entropy chain rule therefore yields
\[
 H(K)\le\left(2+\log\frac1\rho\right)\I(f)
       =\log(4+e)\I(f).
\]
Endpoints follow by continuity.
\end{proof}

\begin{lemma}\label{lem:centering}
For every Boolean $f$,
\[
 \HH(\chi_U f)=\HH(f),\qquad
 \D(f)=\min_U\I(\chi_U f),\qquad
 \D(f)\le4\bigl(\I(f)-\Var(f)\bigr).
\]
Moreover, $\D(\chi_U f)=\D(f)$.
\end{lemma}
\begin{proof}
The coefficient identity
$\wh{\chi_Uf}(S)=\wh f(S\symdiff U)$ proves entropy invariance.
For each $i$, character multiplication either leaves $\Inf_i(f)$ unchanged or replaces it by $1-\Inf_i(f)$.
This proves the minimization and invariance statements.

Put $\alpha_i=\Inf_i(f)$ and $c_i=\wh f(\{i\})$.
Since $|c_i|=|\E d_i f|\le\E|d_i f|=\alpha_i$,
\[
 \min\{\alpha_i,1-\alpha_i\}
 \le2\alpha_i(1-\alpha_i)\le2(\alpha_i-c_i^2).
\]
Summing and using $k\le2(k-1)$ for $k\ge2$ gives
\[
 \D(f)\le2\sum_{|S|\ge2}|S|\wh f(S)^2
 \le4\sum_{|S|\ge2}(|S|-1)\wh f(S)^2
 =4\bigl(\I(f)-\Var(f)\bigr).
\]
\end{proof}

The factor $4$ in the last universal comparison is sharp: for the conjunction of two coordinates, $\D=\I=1$ and $\Var=3/4$.

\section{Antichains and spectral ranks}\label{sec:rank}

\begin{lemma}\label{lem:pivotal-antichain}
Let $f$ be monotone, and let every coordinate in $U_i$ be at least as desirable as $i$.
On each fiber obtained by fixing the coordinates outside $U_i\cup\{i\}$, the support of $d_i f$ is an antichain in the cube on $U_i$.
\end{lemma}
\begin{proof}
Suppose that $y<y'$ are two pivotal assignments on $U_i$, and choose a coordinate $j$ that changes from $-1$ to $1$ between them.
Pivotality gives $f(y,i=1)=1$.
Desirability gives $f(y^{j=1},i=-1)\ge f(y,i=1)=1$.
Monotonicity then gives $f(y',i=-1)=1$, contradicting pivotality at $y'$.
\end{proof}

\begin{lemma}\label{lem:antichain-density}
An antichain in $\bits^q$ has density at most
\[
 c_q=2^{-q}\binom q{\lfloor q/2\rfloor}\le(q+1)^{-1/2}.
\]
\end{lemma}
\begin{proof}
Sperner's theorem \cite{Sperner28} bounds its size by $\binom q{\lfloor q/2\rfloor}$.
For $q=2r$, put $b_r=4^{-r}\binom{2r}r$.
Since $b_0=1$ and $b_{r+1}/b_r=(2r+1)/(2r+2)$, induction gives
$b_r\le(2r+1)^{-1/2}$; the induction step is
$(2r+1)(2r+3)\le(2r+2)^2$.
For $q=2r+1$, the density is $b_r(2r+1)/(2r+2)\le(2r+2)^{-1/2}$.
\end{proof}

\begin{lemma}[Relative noise estimate]\label{lem:antichain-noise}
Let $d:\bits^N\to\{0,\pm1\}$, let $\alpha=\E d^2>0$, and suppose that the support of $d$ is an antichain on every fiber of a set $U$ of $r-1$ coordinates, in a fixed orientation.
Then, for $0<t\le1$,
\begin{equation}\label{eq:relative-noise}
 \sum_T\wh d(T)^2(1-t)^{|T\cap U|}
 \le\alpha\,\E c_B\le\frac{\alpha}{\sqrt{rt}},
 \qquad B\sim\operatorname{Bin}(r-1,t).
\end{equation}
\end{lemma}
\begin{proof}
Resample the coordinates of $U$ at rate $t$, and denote the resampling mask by $R$.
Given $X$ and the resampling mask $R$, the assignments for which $d(Y)\ne0$ form an antichain in a cube of dimension $|R|$.
Lemma~\ref{lem:antichain-density} bounds their conditional probability by
$c_{|R|}\le(|R|+1)^{-1/2}$.
Since $|d|=d^2$ and $R$ is independent of $X$,
\[
 \E[d(X)d(Y)]
 \le\E[|d(X)|\,|d(Y)|]
 \le\alpha\,\E c_{|R|}
 \le\alpha\,\E(|R|+1)^{-1/2}.
\]
Here $|R|\sim\operatorname{Bin}(r-1,t)$, and
\[
 \E\frac1{|R|+1}
 =\int_0^1(1-t+ts)^{r-1}\,ds
 =\frac{1-(1-t)^r}{rt}\le\frac1{rt}.
\]
Cauchy--Schwarz and \eqref{eq:partialnoise} finish the proof.
\end{proof}

The lemma applies to arbitrary signs on any subset of a pivotal antichain. This point will allow both Boolean combinations and character multiplication.

\begin{lemma}[Rank tail and mean]\label{lem:ranktail}
Under the hypotheses of Lemma~\ref{lem:antichain-noise}, let
\[
 \nu(T)=\frac{\wh d(T)^2}{\alpha},
 \qquad
 F(T)=\ln\frac r{|T\cap U|+1}.
\]
Then $F\ge0$ and, for $u\ge0$,
\begin{equation}\label{eq:ranktail}
 \nu(F>u)\le\min\{1,\sqrt{2e}\,e^{-u/2}\},
 \qquad
 \E_\nu F\le\Lp=1+\ln2.
\end{equation}
\end{lemma}
\begin{proof}
Put $K=re^{-u}$ and $\kappa=|T\cap U|$.
If $K\le1$, the event is empty.
Otherwise $F>u$ implies $\kappa<K-1$.
Apply \eqref{eq:relative-noise} with $t=1/(2K)$ to obtain
\[
 \nu(F>u)
 \le(1-1/(2K))^{-(K-1)}\sqrt{\frac{2K}{r}}
 \le\sqrt{2e}\,e^{-u/2}.
\]
Indeed,
\[
 -(K-1)\ln(1-1/(2K))
 \le\frac{K-1}{2K-1}\le\frac12.
\]
For the mean, retain the exact densities in \eqref{eq:relative-noise}.
Put $M=r-1$ and let $h_k=\sum_{j=1}^k1/j$, with $h_0=0$.
The integral representation of the harmonic numbers gives
\begin{align*}
 \E_\nu h_\kappa
 &=\int_0^1\frac{1-\E_\nu(1-t)^\kappa}{t}\,dt\\
 &\ge\sum_{q=1}^M(1-c_q)\binom Mq
                  \int_0^1t^{q-1}(1-t)^{M-q}\,dt
 =\sum_{q=1}^M\frac{1-c_q}{q}.
\end{align*}
Integral comparison gives
$\ln((M+1)/(k+1))\le h_M-h_k$ for $0\le k\le M$.
Consequently,
\begin{equation}\label{eq:harmonic-rank}
 \E_\nu F\le\sum_{q=1}^M\frac{c_q}{q}\le1+\ln2.
\end{equation}
To evaluate the final bound, put $b_j=4^{-j}\binom{2j}j$.
Then $c_{2j-1}=c_{2j}=b_j$, and the binomial series gives
\[
 \sum_{j\ge1}\frac{b_j}{2j-1}=1,
 \qquad \sum_{j\ge1}\frac{b_j}{2j}=\ln2.
\]
The first identity follows by evaluating the series for $1-\sqrt{1-z}$
at $z=1$; the second follows by integrating
$((1-z)^{-1/2}-1)/z$ from $0$ to $1$.
For $r=1$, the sum in \eqref{eq:harmonic-rank} is empty and $F=0$.
\end{proof}

\section{Regular functions and Boolean combinations}\label{sec:regular}

Fix an oriented desirability order for each $h_\ell$, write $r_\ell(i)$ for the rank of $i$, and let $U_i^\ell$ be its predecessors.
Irrelevant coordinates may be placed last.
For $f=F(h_1,\ldots,h_m)$ and $\alpha_i=\Inf_i(f)>0$, partition the pivotal assignments of $i$ according to the first component function for which $i$ is pivotal.
Thus
\begin{equation}\label{eq:pieces}
 d_i f=\sum_{\ell=1}^m d_{i,\ell},\qquad
 \alpha_i=\sum_{\ell=1}^m\alpha_{i,\ell},
 \qquad \alpha_{i,\ell}=\|d_{i,\ell}\|_2^2,
\end{equation}
and the pieces have disjoint supports.
Each nonzero piece is supported on pivotal antichains for $h_\ell$.

\begin{lemma}\label{lem:multi-rank}
For
\[
 G_i(T)=\min_{1\le\ell\le m}
       \ln\frac{r_\ell(i)}{|T\cap U_i^\ell|+1},
\]
the derivative law of $f$ satisfies
\begin{equation}\label{eq:multirank}
 \E_{\nu_i}G_i\le\Lmix+2\ln m,
 \qquad \Lmix=2+\gamma.
\end{equation}
The same conclusion holds if all the pieces in \eqref{eq:pieces} are multiplied by an arbitrary fixed coordinate character on the derivative cube.
\end{lemma}
\begin{proof}
Omit zero pieces and put
$\nu_{i,\ell}(T)=\wh{d_{i,\ell}}(T)^2/\alpha_{i,\ell}$.
Weighted Cauchy--Schwarz gives the pointwise domination
\[
 \nu_i(T)
 =\frac{|\sum_\ell\wh{d_{i,\ell}}(T)|^2}{\alpha_i}
 \le\sum_\ell\frac{\wh{d_{i,\ell}}(T)^2}{\alpha_{i,\ell}}
 =\sum_\ell\nu_{i,\ell}(T).
\]
Write $\kappa_\ell(T)=|T\cap U_i^\ell|$ and
\[
 Z(T)=e^{-G_i(T)}
     =\max_\ell\frac{\kappa_\ell(T)+1}{r_\ell(i)}.
\]
For $t>0$, apply \eqref{eq:relative-noise} to the $\ell$th piece
with resampling rate $1-e^{-t/r_\ell(i)}$. It gives
\[
 \E_{\nu_{i,\ell}}e^{-t(\kappa_\ell+1)/r_\ell(i)}
 \le\frac{e^{-t/r_\ell(i)}}
          {\sqrt{r_\ell(i)(1-e^{-t/r_\ell(i)})}}
 \le t^{-1/2},
\]
where the last inequality uses $1-e^{-u}\ge ue^{-u}$.
Since $Z\ge(\kappa_\ell+1)/r_\ell(i)$ and
$\nu_i\le\sum_\ell\nu_{i,\ell}$,
\[
 \E_{\nu_i}e^{-tZ}\le\min\{1,mt^{-1/2}\}.
\]
Frullani's identity and $0<Z\le1$ now yield
\begin{align*}
 \E_{\nu_i}G_i
 &=\int_0^\infty\frac{\E_{\nu_i}e^{-tZ}-e^{-t}}t\,dt\\
 &\le\int_0^{m^2}\frac{1-e^{-t}}t\,dt
      +\int_{m^2}^\infty\frac{mt^{-1/2}-e^{-t}}t\,dt
 =2\ln m+\gamma+2.
\end{align*}
Here we used
$\int_0^A(1-e^{-t})\,dt/t-\int_A^\infty e^{-t}\,dt/t=\ln A+\gamma$.
This proves \eqref{eq:multirank}.
Character multiplication preserves the disjoint supports, their squared norms, and their antichain property, so the same proof applies.
\end{proof}

\subsection{A reference for one coordinate order}

Relabel a single coordinate order as $1,\ldots,n$.
Euler's product for the hyperbolic sine \cite[Eq.~4.36.1]{DLMF} gives
\[
 \prod_{i\ge1}\left(1+\frac z{i^2}\right)
 =\frac{\sinh(\pi\sqrt z)}{\pi\sqrt z}
 =\sum_{k\ge0}\frac{\pi^{2k}}{(2k+1)!}z^k.
\]
Comparison of coefficients shows that
\begin{equation}\label{eq:euler-reference}
 Q_k(S)=\frac{(2k+1)!}{\pi^{2k}}\prod_{i\in S}i^{-2},
 \qquad |S|=k,
\end{equation}
is a subprobability reference on the $k$-element subsets of $[n]$.
For $S=\{i_1<\cdots<i_k\}$, put
$F_{i_j}(S\setminus\{i_j\})=\ln(i_j/j)$.
Then
\begin{align}
 -\ln Q_k(S)
 &=2\sum_{j=1}^kF_{i_j}(S\setminus\{i_j\})
     +\ln\frac{\pi^{2k}(k!)^2}{(2k+1)!}\notag\\
 &\le2\sum_{j=1}^kF_{i_j}(S\setminus\{i_j\})
     +k\ln\frac{\pi^2}{4}.
 \label{eq:euler-cost}
\end{align}
The inequality $(2k+1)!\ge4^k(k!)^2$ follows by induction: the quotient
at $k+1$ divided by the quotient at $k$ is $(2k+3)/(2k+2)>1$.
The empty set has reference mass one.

\subsection{Reference measures on a Fourier level}

For $k\ge1$ and $a>1$, put
\begin{equation}\label{eq:rankweights}
 \vartheta_k(i)=\max_{1\le\ell\le m}
                  \min\{1,(k/r_\ell(i))^a\}.
\end{equation}
Give each $k$-element set $S$ weight $\prod_{i\in S}\vartheta_k(i)$ and normalize by $Z_k$.
Each rank order is a permutation, so integral comparison gives
\[
 \sum_i\min\{1,(k/r_\ell(i))^a\}\le\frac{ak}{a-1}.
\]
Consequently,
\begin{equation}\label{eq:normalizer}
 Z_k\le\frac{(\sum_i\vartheta_k(i))^k}{k!}
       \le\left(\frac{ema}{a-1}\right)^k.
\end{equation}
If $i\in S$ and $T=S\setminus\{i\}$, then
\[
 \ln\frac1{\vartheta_{|S|}(i)}
 =a\min_\ell\ln^+\frac{r_\ell(i)}{|S|}
 \le aG_i(T).
\]
Averaging under the Fourier law and using \eqref{eq:derivlaw} gives
\begin{equation}\label{eq:rankcost}
 \E\sum_{i\in\mathcal S}\ln\frac1{\vartheta_{|\mathcal S|}(i)}
 \le a(\Lmix+2\ln m)\I(f).
\end{equation}
Gibbs' inequality on each level, \eqref{eq:normalizer}, \eqref{eq:rankcost}, and Lemma~\ref{lem:levelentropy} prove
\begin{equation}\label{eq:ordinaryexact}
 \HH(f)\le
 \left(\Clev+\frac{1+\ln\frac a{a-1}+a\Lmix+(1+2a)\ln m}{\ln2}\right)\I(f).
\end{equation}
For $a=6/5$, the constant independent of $m$ is
$11.2375133696\ldots$, which proves \eqref{eq:multiI}.

\subsection{Centering the influence}

Let
\[
 U=\{i:\Inf_i(f)>1/2\},\qquad b=|U|,\qquad g=\chi_U f.
\]
Then $\I(g)=\D(f)$ and $\HH(g)=\HH(f)$.
For $i\notin U$, the derivative of $g$ is $\chi_Ud_i f$.
Lemma~\ref{lem:multi-rank} therefore applies directly to its signed pieces.

For the reference law, replace $\vartheta_k(i)$ by $1$ when $i\in U$, and leave it unchanged otherwise.
No rank cost is charged to a coordinate in $U$.
Its normalizing constant is at most
\[
 \left(e\left(b+\frac{ma}{a-1}\right)\right)^k,
\]
because $k\ge1$.
Repeating the preceding entropy argument, now for $g$, yields
\begin{equation}\label{eq:centerexact}
 \HH(f)\le
 \left(\Clev+\frac{1+\ln\!\left(b+\frac{ma}{a-1}\right)
                    +a(\Lmix+2\ln m)}{\ln2}\right)\D(f).
\end{equation}

To bound $b$, observe that for a unate function $h$,
$|\wh h(\{i\})|=\Inf_i(h)$, whence
$\sum_i\Inf_i(h)^2\le1$.
An edge on which $f$ changes changes at least one component function. Thus
\[
 \sum_i\Inf_i(f)^2
 \le m\sum_{\ell=1}^m\sum_i\Inf_i(h_\ell)^2
 \le m^2,
\]
and hence $b<4m^2$.

For $m=1$, there is a sharper observation.
If a coordinate has a predecessor in a desirability order, its pivotal support on that one-dimensional predecessor fiber is an antichain.
Its conditional pivotal probability is at most $1/2$, so its influence is at most $1/2$.
Every coordinate in $U$ therefore has rank $1$ and has rank cost zero.
For $i\notin U$, the derivative of $g$ is a signed function with the same
antichain supports as $d_i f$.
Lemma~\ref{lem:ranktail} and the derivative law therefore give
\[
 \E\sum_{i\in\mathcal S}F_i(\mathcal S\setminus\{i\})
 \le\Lp\I(g),
\]
where $\mathcal S$ has the spectral law of $g$.
Use the reference \eqref{eq:euler-reference} for this law.
Gibbs' inequality on each level, \eqref{eq:euler-cost}, and
Lemma~\ref{lem:levelentropy} yield
\begin{align*}
 \HH(f)=\HH(g)
 &\le\left(\Clev+\frac{2\Lp+\ln(\pi^2/4)}{\ln2}\right)\I(g)\\
 &=\bigl(\Clev+2\log(e\pi)\bigr)\D(f)=\Creg\D(f).
\end{align*}
This proves Theorem~\ref{thm:A}, including its character and variance-deficit consequences by Lemma~\ref{lem:centering}.

For $m\ge2$, the integer bound $b\le4m^2-1$ gives
\[
 b+6m\le7m^2,
\]
since $3m^2-6m+1>0$.
The same bound is harmless for $m=1$, for which $b\le1$.
Taking $a=6/5$ in \eqref{eq:centerexact} proves \eqref{eq:multiD}, with the exact constant
\begin{equation}\label{eq:cm}
 \Cmulti=\Clev+\frac{1+\ln7+\frac65(2+\gamma)}{\ln2}
               +\frac{22}{5}\log m
       <11.460+\frac{22}{5}\log m.
\end{equation}
Lemma~\ref{lem:centering} proves the remaining assertions of Theorem~\ref{thm:B}.
If $\D(f)=0$, the centered function is constant and all these inequalities hold with both sides zero.

\subsection{Several desirability chains}

A \emph{desirability chain partition} is a partition of the coordinates into ordered lists such that every earlier coordinate in a list is at least as desirable as every later one.
The orientation is fixed for the whole function.
Singleton lists always give such a partition.

\begin{theorem}\label{thm:width}
Let $w\ge1$. If a unate Boolean function $f$ has a desirability chain partition with at most $w$ lists, then
\begin{equation}\label{eq:widthsingle}
 \HH(f)\le(\Creg+\log w)\D(f).
\end{equation}
If each of $h_1,\ldots,h_m$ has such a partition with at most $w$ lists, then $f=F(h_1,\ldots,h_m)$ satisfies
\[
 \HH(f)\le\left(11.238+\frac{17}{5}\log m+\log w\right)\I(f)
\]
and
\[
 \HH(f)\le\left(11.460+\frac{22}{5}\log m+\log w\right)\D(f).
\]
\end{theorem}
\begin{proof}
Use rank within a list and predecessors in that list.
All the antichain and rank estimates remain unchanged.
For a single component, condition on the total Fourier-set size $k$ and
let $k_1,\ldots,k_w$ be its counts in the respective lists.
Use the multinomial reference
\[
 q_k(k_1,\ldots,k_w)=\frac{k!}{\prod_{\ell=1}^wk_\ell!}\,w^{-k}
\]
for these counts, and the product of the references
\eqref{eq:euler-reference} within the lists.
This is a subprobability reference. The count-vector cost is at most
$k\ln w$, and \eqref{eq:euler-cost} bounds its remaining cost by
twice the sum of the within-list rank costs plus $k\ln(\pi^2/4)$.
High-influence coordinates have rank $1$ in their respective lists and
therefore have zero rank cost after centering.
The argument for Theorem~\ref{thm:A} proves \eqref{eq:widthsingle}.

For several components, each order system satisfies
\[
 \sum_i\min\{1,(k/r_\ell(i))^a\}\le\frac{wak}{a-1}.
\]
This adds $\log w$ to \eqref{eq:ordinaryexact}.
The centered normalizer is bounded using
$b+6mw\le7m^2w$.
For $m\ge2$ this follows from $b\le4m^2-1$ and $w\ge1$.
For $m=1$ one can use the preceding sharper argument directly.
The rest of the proof is unchanged.
\end{proof}

\subsection{Halfspaces, score transitions, and restrictions}

\begin{corollary}\label{cor:scores}
Let $z(x)=\sum_iw_ix_i$, and let $f(x)$ be a Boolean labeling of the attained values of $z$ with at most $m\ge1$ transitions in their increasing order.
Then the centered and character bounds of Theorem~\ref{thm:B} hold.
The ordinary bound improves to
\[
 \HH(f)\le\left(11.238+\frac{12}{5}\log m\right)\I(f).
\]
\end{corollary}
\begin{proof}
Place a threshold between each pair of consecutive attained scores whose labels differ.
Then $f$ is a Boolean function of the resulting $m$ thresholds.
All these gates have the same order by decreasing $|w_i|$.
Thus the reference \eqref{eq:rankweights} needs only one order, removing the $\ln m$ term from its normalizer in \eqref{eq:ordinaryexact}.
\end{proof}

A union of $r$ intervals or rays has at most $2r$ effective transitions on the attained scores.
Thus the corollary includes finite unions of parallel slabs, with arbitrary endpoint conventions and no separation condition.
Constant labelings have zero Fourier entropy.

The classes in Theorems~\ref{thm:A}, \ref{thm:B}, and~\ref{thm:width} are closed under coordinate restrictions.
As another consequence, if fixing $k$ specified coordinates makes every restriction regular, the function is a Boolean combination of those $k$ coordinate functions and its $2^k$ regular restrictions.
Theorem~\ref{thm:B} gives a constant of order $k+1$ for this class.

\section{Read-\texorpdfstring{$k$}{k} unions of structured gates}\label{sec:union}

A union can contain arbitrarily many component functions while each input is relevant to only a few of them.
We show that this overlap parameter controls Fourier entropy.
In this section a union is true, with value $1$, when at least one component has value $1$.
Intersections follow by complementing the components and the output.

\begin{theorem}\label{union:theorem}
Let $f=h_1\vee\cdots\vee h_N$ or $f=h_1\wedge\cdots\wedge h_N$.
Suppose each $h_l$ is a Boolean function of at most $M$ regular functions, and every input coordinate is relevant to at most $k$ of the functions $h_l$, where $k,M\ge1$.
Then
\begin{equation}\label{union:ordinary}
 \HH(f)\le\left(41+7\log k+\frac{29}{5}\log M\right)\I(f).
\end{equation}
Put
\[
 \HH^+(f)=\sum_{A\ne\varnothing}\wh f(A)^2\log\frac1{\wh f(A)^2},
 \qquad
 J(f)=\sum_{|A|\ge2}|A|\wh f(A)^2.
\]
The strengthened estimate
\begin{equation}\label{union:strong}
 \HH^+(f)\le\left(66+\frac{59}{5}\log k+\frac{53}{5}\log M\right)J(f)
\end{equation}
also holds. Since $J(f)\le2(\I(f)-\Var(f))$, this implies FEI+ with coefficient
$132+\frac{118}{5}\log k+\frac{106}{5}\log M$.
If each inner function instead has a desirability chain partition with at most $w$ chains, add $\log w$ to the coefficients in \eqref{union:ordinary} and \eqref{union:strong}.
\end{theorem}

In particular, read-$k$ unions and intersections of any number of halfspaces satisfy
$\HH(f)\le(41+7\log k)\I(f)$.
The theorem also includes read-$k$ unions of polytopes with at most $M$ facets and read-$k$ DNF/CNF formulas.
A Boolean function on $w$ variables is a Boolean combination of its $w$ coordinate functions, so arbitrary gates of width $w$ can be included by taking $M=w$.
Here a variable is counted once per outer gate, regardless of how many inner gates use it.

We first prove \eqref{union:ordinary} for regular inner functions.
Let $R_l$ be the relevant support of $h_l$.
Its regular representation may be assumed to use only $R_l$: fix all irrelevant variables identically in its inner functions, which preserves regularity and the represented function.
For a union, remove constant false gates; a constant true gate makes the assertion immediate.
Thus write
\[
 p_0=\P[f=-1],\qquad p_l=\P[h_l=1]>0,\qquad
 \pi_l=\min\{p_l,1/2\},\qquad b_l=\log(1/\pi_l).
\]

\subsection{The path bound and rare pivotal supports}

\begin{lemma}\label{union:path}
For arbitrary outer gates with the stated overlap bound,
\[
 p_0\sum_l p_l\le\frac{k}{4}\I(f).
\]
\end{lemma}
\begin{proof}
For fixed $l$, independently resample all coordinates in $R_l$, one at a time in a fixed order.
The event that $f$ is initially false is independent of the event that $h_l$ is true after this resampling.
With probability $p_0p_l$ the path therefore goes from a false to a true input, so some step is a false-to-true transition.
Every intermediate input is uniform, and resampling coordinate $j$ makes that transition with probability $\Inf_j(f)/4$.
A union bound, followed by summation over $l$, proves the claim.
\end{proof}

The following fixed-noise estimate will pay for selecting rare gates.
It applies to arbitrary Boolean gates and arbitrary signed pieces.

\begin{lemma}\label{union:rare}
Let $h$ have support $R\ni j$, set $\beta=120/11$ and $b=\log(1/\min\{\P[h=1],1/2\})$, and let $d\in\{0,\pm1\}$ vanish outside the $j$-pivotal support of $h$.
Suppose $\|d\|_2>0$ and write $\nu(T)=\wh d(T)^2/\|d\|_2^2$.
Then
\begin{equation}\label{union:raretail}
 \nu\bigl(|T\cap R|+1<b/\beta\bigr)
 \le 2^{3/5}\,2^{-5b/12}.
\end{equation}
\end{lemma}
\begin{proof}
The event is empty if $b\le\beta$.
Otherwise $\P[h=1]=2^{-b}$ and $\Inf_j(h)\le2^{1-b}$.
Fix only the coordinates outside $R$, obtaining a fiber $d_z$ on $R\setminus\{j\}$.
Its support density $\alpha_z$ is at most $2^{1-b}$.
The hypercontractive inequality \cite{Bonami70,Beckner75,OD14} at noise correlation $\rho=1/4$ gives
\[
 \|T_{\sqrt{\rho}}d_z\|_2^2
 \le\|d_z\|_{1+\rho}^2=\alpha_z^{8/5}.
\]
Average over the outside fibers, apply Parseval there, and divide by $\E_z\alpha_z$.
It follows that
\[
 \E_{T\sim\nu}\rho^{|T\cap R|}
 \le 2^{(3/5)(1-b)}.
\]
On the event in \eqref{union:raretail}, $|T\cap R|<b/\beta-1\le b/\beta$.
Thus its probability is at most
$\rho^{-b/\beta}2^{(3/5)(1-b)}=2^{3/5}2^{-5b/12}$, since $3/5-2/\beta=5/12$.
\end{proof}

\subsection{A reference measure generated by a queue}

Fix
\[
 a=\frac65,\qquad \beta=\frac{120}{11},\qquad \lambda=\frac1k,\qquad B_*=\sqrt{2e}.
\]
For each regular inner gate $h_{l,r}$, order its relevant coordinates by desirability.
Write $r_{l,r}(j)$ for the rank and $U_j^{l,r}$ for the predecessors of $j$.
For $s\ge1$ and $i\in R_l$, set
\[
 \theta_s^l(i)=\max_{r:\,i\in\supp(h_{l,r})}
                  \min\{1,(s/r_{l,r}(i))^a\}.
\]
The relevant supports of the inner functions cover $R_l$, so this maximum is nonempty.
Let $Z_s^l$ be the elementary symmetric polynomial of degree $s$ in these weights.
As in the rank reference of Section~\ref{sec:rank},
\begin{equation}\label{union:normalizer}
 \sum_{i\in R_l}\theta_s^l(i)\le\frac{Msa}{a-1},
 \qquad
 Z_s^l\le\left(\frac{eMa}{a-1}\right)^s.
\end{equation}

We define a subprobability law $Q_1$ on subsets of the input coordinates.
Independently choose each gate $l$ as a root with probability
$\lambda\pi_l/(1+\lambda\pi_l)$, and put the roots in a queue in a fixed order.
For each occurrence of a gate $l$ removed from the queue, perform these choices:
\begin{enumerate}[label=(\roman*)]
\item Choose $s\ge1$ with probability $1/[s(s+1)]$.
If $s>|R_l|$, the run fails; otherwise choose an $s$-subset $B\subseteq R_l$ with probability
$\prod_{i\in B}\theta_s^l(i)/Z_s^l$.
\item Choose $c\ge0$ with probability $1/[(c+1)(c+2)]$, and choose a multiset of $c$ elements of $B$ uniformly.
For each occurrence of a selected variable $j$, choose uniformly one of the at most $k$ gates containing $j$ and append it to the queue.
Process the multiset occurrences in a fixed coordinate order, with repeated occurrences ordered consecutively.
\end{enumerate}
If the queue becomes empty, output the union of all selected subsets.
Let $Q_1(A)$ be the probability of a successful finite output equal to $A$.
A run may fail or continue forever; its mass is then absent from $Q_1$.
Consequently $\sum_AQ_1(A)\le1$, which is all that Gibbs' inequality requires.
Repeated visits are treated as distinct occurrences in the random process.

\subsection{A witness run and its cost}

Fix $A\ne\varnothing$ with $\wh f(A)\ne0$, and put $\kappa_l=|A\cap R_l|$.
Call $l$ dense if $\kappa_l\ge b_l/\beta$.
For each $j\in A\cap R_l$, put
\[
 F_l^j(A)=\min_{r:\,j\in\supp(h_{l,r})}
          \ln\frac{r_{l,r}(j)}{|A\cap U_j^{l,r}|+1},
 \qquad
 G_j(l)=\frac a{\ln2}F_l^j(A)+b_l\mathbf1[l\text{ is not dense}].
\]
Every $j\in A$ belongs to some $R_l$.
Choose $\tau(j)$ minimizing $G_j(l)$, with a fixed tie rule, and set
\[
 A_l=\{j\in A:\tau(j)=l\},\qquad
 \mathcal F=\{l:A_l\ne\varnothing\},\qquad
 G_j=G_j(\tau(j)).
\]
The nonempty $A_l$ partition $A$.

Consider the directed graph on $\mathcal F$ having an edge $l\to l'$ when
$l\ne l'$ and $A_l\cap R_{l'}\ne\varnothing$.
Choose one representative from each source strongly connected component as a root.
Every vertex is reachable from these roots, so choose a directed spanning forest, and attach each nonroot $l'$ to its parent $l$ through an element of $A_l\cap R_{l'}$.
This gives a finite run producing $A$: prescribe these roots, choose $B_l=A_l$, and prescribe exactly the forest children.
Every gate in $\mathcal F$ is visited once.
If $c_l$ is its number of children, then
\begin{equation}\label{union:forestcount}
 |\mathrm{roots}|+\sum_{l\in\mathcal F}c_l=|\mathcal F|\le|A|.
\end{equation}

The sets $A\cap R_r$ for distinct roots $r$ are pairwise disjoint.
Indeed, if $j$ belonged to two such sets and $l=\tau(j)$, then $l$ has an edge to each root except when it equals that root.
Since the root components receive no edges from other components, $l$ would lie in both root components, a contradiction.
This includes the case where $l$ itself is one of the roots.

The root selection costs
\[
 \Lambda+\sum_{\mathrm{roots}\ r}(\log k+b_r),
 \qquad
 \Lambda=\sum_l\log(1+\lambda\pi_l)
 \le\frac{\log e}{k}\sum_l\pi_l.
\]
By the disjointness just proved, dense-root rarity costs sum to at most $\beta|A|$.
Each nondense root owns at least one variable, so its rarity cost can be paid from
$\sum_{j\in A}b_{\tau(j)}\mathbf1[\tau(j)\text{ not dense}]$.

For $i\in A_l$, set $s_l=|A_l|$ and choose an inner order attaining $F_l^i(A)$.
Because $|A\cap U_i^{l,r}|+1\le\kappa_l$,
\[
 \ln^+\frac{r_{l,r}(i)}{s_l}
 \le F_l^i(A)+\ln\frac{\kappa_l}{s_l}.
\]
The log-sum inequality and the read bound give
\begin{equation}\label{union:logsum}
 \sum_{l\in\mathcal F}s_l\log\frac{\kappa_l}{s_l}
 \le |A|\log\frac{\sum_{l\in\mathcal F}\kappa_l}{|A|}
 \le |A|\log k.
\end{equation}
Thus \eqref{union:normalizer} bounds the subset-selection cost by
\[
 |A|\log\frac{eMa}{a-1}
 +a|A|\log k+\frac a{\ln2}\sum_{j\in A}F_{\tau(j)}^j(A).
\]

The remaining choices have elementary costs.
The choices of $s_l$ cost at most $|A|\log3$, since $s(s+1)\le3^s$ for $s\ge1$.
The choices of $c_l$ cost at most $|A|(1+\log3)$, since $(c+1)(c+2)\le2\,3^c$.
There are $\binom{s+c-1}{c}\le2^{s+c}$ multisets, so their total cost is at most $2|A|$.
Finally, root and child gate labels together cost at most $|A|\log k$, by \eqref{union:forestcount}.
Putting
\begin{equation}\label{union:K0}
 K_0=\frac{120}{11}+3+2\log3+\log(6e),
\end{equation}
we conclude that
\begin{equation}\label{union:runbound}
 \log\frac1{Q_1(A)}
 \le \Lambda+|A|\left(K_0+\frac{11}{5}\log k+\log M\right)
               +\sum_{j\in A}G_j.
\end{equation}

\subsection{Averaging the rank and rarity costs}

For a coordinate $j$ with $\alpha_j=\Inf_j(f)>0$, let
$\nu_j(T)=\wh{d_jf}(T)^2/\alpha_j$.
If $j$ is pivotal for the union, it is pivotal for some component containing $j$.
Split $d_jf$ into disjoint signed pieces according to the least such component.
There are at most $k$ nonzero pieces.
Their normalized Fourier laws satisfy
\[
 \nu_j(T)\le\sum_l\nu_{j,l}(T)
\]
by weighted Cauchy--Schwarz, as in Lemma~\ref{lem:multi-rank}.

For a specified piece, a further split among the at most $M$ regular inner gates and Lemma~\ref{lem:ranktail} give
\[
 \nu_{j,l}(F_l^j>v)\le MB_*e^{-v/2}.
\]
Lemma~\ref{union:rare} gives
\[
 \nu_{j,l}(l\text{ is not dense})\le2^{3/5}2^{-5b_l/12}.
\]
Since $G_j\le G_j(l)$ and $1/(4a)=5/24$, splitting the rank and rarity costs at $u/2$ gives the same exponent for both terms. The union bound therefore yields
\begin{equation}\label{union:Gtail}
 \nu_j(G_j>u)
 \le\min\{1,\ k(MB_*+2^{3/5})2^{-5u/24}\}
 \le\min\{1,\ kM(B_*+2^{3/5})2^{-5u/24}\}.
\end{equation}
Here all the statistics are evaluated at $A=T\cup\{j\}$.
Integrating the tail, and setting
\begin{equation}\label{union:E0}
 E_0=\frac{24}{5}\left[\log(\sqrt{2e}+2^{3/5})+\log e\right],
\end{equation}
gives
\begin{equation}\label{union:Gmean}
 \E_{\nu_j}G_j\le E_{k,M}:=E_0+\frac{24}{5}\log k+\frac{24}{5}\log M.
\end{equation}

Put $v=\Var(f)$, and use the subprobability reference
$Q(\varnothing)=1-v$, $Q(A)=vQ_1(A)$ for $A\ne\varnothing$.
Gibbs' inequality gives
\[
 \HH(f)\le h_2(v)+
 \sum_{A\ne\varnothing}\wh f(A)^2\log\frac1{Q_1(A)}.
\]
The proof of Lemma~\ref{lem:levelentropy} gives $h_2(v)\le2\I(f)$.
Since $v\le4p_0$, Lemma~\ref{union:path} gives $v\Lambda\le(\log e)\I(f)$.
Finally,
$\wh f(A)^2=\alpha_j\nu_j(A\setminus\{j\})$ for $j\in A$.
Averaging \eqref{union:runbound} and applying \eqref{union:Gmean}, we obtain
\[
 \HH(f)\le
 \left[2+\log e+K_0+E_0+7\log k+\frac{29}{5}\log M\right]\I(f).
\]
The exact constant $2+\log e+K_0+E_0$ is $40.80487606\ldots<41$.
We record rational certificates for this rounding and for the strengthened constant below:
\begin{equation}\label{union:rounding}
 \log e<\frac{1443}{1000},\qquad
 \log3<\frac{317}{200},\qquad
 \log(\sqrt{2e}+2^{3/5})<\frac{243}{125}.
\end{equation}
The exponential series gives
$e<\sum_{j=0}^6 1/j!+1/4410=31967/11760<87/32$.
The first two inequalities follow from $87^{1000}<2^{6443}$ and $3^{200}<2^{317}$.
For the last one, the rational comparisons
\[
 \frac{2\cdot31967}{11760}
   <\left(\frac{2331645}{10^6}\right)^2,\qquad
 8<\left(\frac{1515717}{10^6}\right)^5
\]
give $\sqrt{2e}+2^{3/5}<1539/400$, and
$1539^{125}<2^{243}400^{125}$ completes the certificate.
Expanding the exact ordinary constant and applying \eqref{union:rounding} gives
\[
 \frac{120}{11}+6+3\log3+\frac{34}{5}\log e
       +\frac{24}{5}\log(\sqrt{2e}+2^{3/5})
 <\frac{2244423}{55000}<41.
\]
This proves \eqref{union:ordinary}.

\subsection{The strengthened entropy estimate}

We give the reduction to nonsingleton Fourier sets explicitly.
For an arbitrary Boolean function, put
\[
 P(A)=\wh f(A)^2,\quad
 W=\sum_{|A|\ge2}P(A),\quad
 W_1=\sum_iP(\{i\}),\quad
 J=\sum_{|A|\ge2}|A|P(A).
\]
Then $\I=W_1+J$, $\Var=W_1+W$, $W\le J/2$, and
\begin{equation}\label{union:Jderiv}
 J=\sum_i\left(\Inf_i(f)-\wh f(\{i\})^2\right)
   =\sum_i\Var(d_if).
\end{equation}

\begin{lemma}\label{union:smallmass}
For every Boolean function,
\[
 \sum_iP(\{i\})\log\frac1{P(\{i\})}\le2(\log e)J,
 \qquad
 W\log\frac1W\le8J.
\]
\end{lemma}
\begin{proof}
For $b=|\wh f(\{i\})|$, we have
$b\le\Inf_i(f)$ and hence
$\Var(d_if)\ge b-b^2$.
Since $b^2\ln(1/b)\le b-b^2$, the first assertion follows.

For the second, let $p=\min\{\P(f=1),\P(f=-1)\}$.
The pivotal bound $\Inf_i(f)\le2p$ gives
\[
 W_1\le\sum_i\Inf_i(f)^2\le2p\,\I(f),
 \qquad J\ge(1-2p)\I(f).
\]
Also
\[
 J\ge\sum_i\Inf_i(f)(1-\Inf_i(f)),
 \quad W\le\Var(f)\le4p,
 \quad \I(f)\ge2p\log(1/p).
\]
If $p\le1/4$, then $J\ge p\log(1/p)$.
When $4p\le1/e$, monotonicity of $x\log(1/x)$ on $(0,1/e]$ gives
$W\log(1/W)\le4p\log(1/(4p))\le4J$.
When $4p>1/e$, we have $J\ge p>1/(4e)$, while
$x\log(1/x)\le(\log e)/e$; therefore $W\log(1/W)\le4(\log e)J<8J$.

It remains to consider $p>1/4$.
Set $\delta=1/48$.
If $J\ge\delta$ and $W\ge\delta/2$, then
$W\log(1/W)\le(J/2)\log96<4J$.
If $J\ge\delta$ and $W<\delta/2$, monotonicity instead gives
$W\log(1/W)\le(\delta/2)\log96<4\delta\le4J$.

Finally suppose $p>1/4$ and $J<\delta$.
Then $\I>3/4$ and $\I=W_1+J<1+\delta$.
Every influence $\alpha_i$ satisfies
$\alpha_i(1-\alpha_i)<\delta$, so $\alpha_i<2\delta$ or $\alpha_i>1-2\delta$.
The influences of the first kind sum to less than
$\delta/(1-2\delta)<3/4$.
Thus some coordinate $i_*$ has $\alpha_{i_*}>1-2\delta$.
Reverse this input if necessary so that
$b=\wh f(\{i_*\})>0$.
By \eqref{union:Jderiv}, $b^2>1-3\delta$, and hence $b>1-3\delta$.

Write $f=A+x_{i_*}B$, where $A=\E_{i_*}f$ and $B=d_{i_*}f$.
They take values in $\{0,\pm1\}$, and exactly one is zero at each input.
Put $D=1-B\in\{0,1,2\}$ and let $E=\{D\ne0\}$ have density $\eta$.
Then $\eta\le1-b<3\delta$ and
\[
 W\le\|A\|_2^2+\Var(B)
   \le\P[D=1]+\E D^2\le4\eta.
\]
If $\eta=0$, then $W=0$ and the assertion is immediate; assume henceforth that $\eta>0$.
Moreover,
\[
 J\ge\sum_{T\ne\varnothing}(|T|+1)\wh B(T)^2
   \ge\sum_T|T|\wh D(T)^2.
\]
Every edge leaving $E$ has squared derivative of $D$ at least $1/4$.
The edge-isoperimetric inequality therefore gives
$J\ge(\eta/2)\log(1/\eta)$.
Since $4\eta<12\delta=1/4<1/e$,
\[
 W\log(1/W)\le4\eta\log(1/(4\eta))
             \le4\eta\log(1/\eta)\le8J.
\]
\end{proof}

It follows that any subprobability reference $Q$ positive on all nonsingleton spectral sets satisfies
\begin{equation}\label{union:plusreduction}
 \HH^+(f)\le(8+2\log e)J+
       \sum_{|A|\ge2}P(A)\log\frac1{Q(A)}.
\end{equation}
Indeed, apply Gibbs' inequality to $P(A)/W$ on $|A|\ge2$, and use Lemma~\ref{union:smallmass} for its mass term and the singleton contribution.

Use $Q=Q_1$ from the queue construction.
The vacuum term admits the bound
\begin{equation}\label{union:strongvacuum}
 W\Lambda\le(\log e)J.
\end{equation}
To see this, use $W\le J/2$ to obtain
\[
 W\I=W(W_1+J)\le J(W_1/2+W)\le\Var(f)J.
\]
Then Lemma~\ref{union:path}, $\Lambda\le(\log e)\sum_l\pi_l/k$, and
$\Var(f)=4p_0(1-p_0)$ give
$W\Lambda\le(\log e)(1-p_0)J\le(\log e)J$.
The constant-output cases are immediate.

We next strengthen \eqref{union:Gmean} to
\begin{equation}\label{union:Gconditioned}
 \E_{\nu_j}[G_j\mathbf1[T\ne\varnothing]]
 \le2E_{k,M}\,\nu_j(T\ne\varnothing).
\end{equation}
If $\nu_j(T\ne\varnothing)\ge1/2$, this follows directly from \eqref{union:Gmean}.
Otherwise
\[
 \Inf_j(f)\ge
 \frac{\wh f(\{j\})^2}{\Inf_j(f)}
 =\nu_j(\varnothing)>\frac12.
\]
Since at most $k$ outer gates contain $j$, some $h_l$ has
$\Inf_j(h_l)>1/(2k)$.
Its true probability exceeds $1/(4k)$, and some regular inner gate has
$\Inf_j(h_{l,r})>1/(2kM)$.
For a regular function, influences do not increase along a desirability order, and their squares sum to at most $1$ by the unate Chow identity.
Thus $r_{l,r}(j)<4k^2M^2$, and pointwise
\[
 G_j\le G_j(l)
 <a\log(4k^2M^2)+\log(4k)
 =\frac{22}{5}+\frac{17}{5}\log k+\frac{12}{5}\log M
 \le2E_{k,M}.
\]
The last inequality follows from \eqref{union:E0}: $2E_0>10>22/5$, and the coefficients $48/5$ of both logarithms in $2E_{k,M}$ exceed $17/5$ and $12/5$.
This proves \eqref{union:Gconditioned}.

Average \eqref{union:runbound} only over $|A|\ge2$.
The set-size term is now $J$, and \eqref{union:Jderiv} and \eqref{union:Gconditioned} bound the rank and rarity term by $2E_{k,M}J$.
Together with \eqref{union:plusreduction} and \eqref{union:strongvacuum}, this gives
\[
 \HH^+(f)\le
 \left[8+3\log e+K_0+2E_0
       +\frac{59}{5}\log k+\frac{53}{5}\log M\right]J.
\]
The exact constant is $65.94577371\ldots<66$.
The same rational bounds \eqref{union:rounding} certify this:
\[
 \frac{120}{11}+12+3\log3+\frac{68}{5}\log e
       +\frac{48}{5}\log(\sqrt{2e}+2^{3/5})
 <\frac{3627321}{55000}<66.
\]
This proves \eqref{union:strong}.

Finally, for chain partitions of width at most $w$, use rank and predecessors within each chain.
Only \eqref{union:normalizer} changes: its upper bound acquires a factor $w$.
The rank tail remains unchanged, and the rank estimate in the last conditional argument still holds within the chosen chain.
Thus \eqref{union:runbound} gains $|A|\log w$.
After averaging, this adds $\log w$ to both theorem coefficients, as claimed.

\subsection{A bound retaining the joint witness cost}

The queue construction permits an estimate depending on the gates used
to describe the whole Fourier set. Let $d(j)$ be the number of outer
gates containing coordinate $j$, so $d(j)\le k$.
Allow the root odds to be $\lambda\pi_l$ for an arbitrary $\lambda>0$,
and put
\[
 \Lambda_\lambda=\sum_l\log(1+\lambda\pi_l).
\]
Denote the resulting subprobability law by $Q_{1,\lambda}$.
All other choices in the queue construction remain as defined above.

Fix a nonempty Fourier set $A$ with $\wh f(A)\ne0$. An assignment chooses, for every
$j\in A$, one gate $\tau(j)$ containing $j$. Put
\[
 A_l=\{j\in A:\tau(j)=l\},\quad s_l=|A_l|,\quad
 \mathcal F=\{l:s_l>0\},\quad \kappa_l=|A\cap R_l|.
\]
Take a directed spanning forest in the incidence graph on $\mathcal F$,
whose roots are exactly one representative from each source strongly connected component.
For each forest edge $e:l\to l'$, choose an attachment
$j_e\in A_l\cap R_{l'}$. Write $\mathcal R$ for the roots and $c_l$
for the number of children of $l$.
Define the finite witness cost
\begin{align*}
 \mathcal C_\lambda(A;\tau,\mathcal T)
  ={}&\sum_{r\in\mathcal R}\log\frac1{\lambda\pi_r}
       +\sum_{e\in E(\mathcal T)}\log d(j_e)\\
    &+\sum_{l\in\mathcal F}\left[
       \log\bigl(s_l(s_l+1)\bigr)+\log Z_{s_l}^l
       -\sum_{j\in A_l}\log\theta_{s_l}^l(j)\right]\\
    &+\sum_{l\in\mathcal F}\left[
       \log\bigl((c_l+1)(c_l+2)\bigr)
       +\log\binom{s_l+c_l-1}{c_l}\right].
\end{align*}
The notation $\mathcal T$ includes the attachment choices and an ordering
of children sharing an attachment coordinate.

\begin{proposition}[Variational witness bound]\label{prop:joint-witness}
For every such assignment and forest,
\[
 \log\frac1{Q_{1,\lambda}(A)}
 \le\Lambda_\lambda+\mathcal C_\lambda(A;\tau,\mathcal T).
\]
Consequently, for $v=\operatorname{Var}(f)$,
\[
 \HH(f)\le h_2(v)+v\Lambda_\lambda+
 \sum_{\substack{A\ne\varnothing\\\wh f(A)\ne0}}\wh f(A)^2
       \min_{\tau,\mathcal T}\mathcal C_\lambda(A;\tau,\mathcal T).
\]
\end{proposition}
\begin{proof}
Prescribe the root set, the subset $A_l$ at each visited gate, and exactly
the forest children with their attachments. This is a successful finite
run and visits every gate in $\mathcal F$ once.
The root-set probability is
$2^{-\Lambda_\lambda}\prod_{r\in\mathcal R}\lambda\pi_r$.
The remaining factors in the run probability are precisely the subset,
size, child-count, multiset and gate-label probabilities defining
$\mathcal C_\lambda$. Thus the probability of this run is
$2^{-\Lambda_\lambda-\mathcal C_\lambda}$, proving the first bound.
Apply Gibbs' inequality to
$Q(\varnothing)=1-v$ and
$Q(A)=vQ_{1,\lambda}(A)$ for $A\ne\varnothing$.
Minimization is valid because each finite witness separately gives a
lower bound for the same reference mass $Q_{1,\lambda}(A)$.
\end{proof}

The variational minimum must be taken over descriptions of $A$ as a
whole. Minimizing one coordinate at a time need not minimize it.

For comparison with the previous estimates, define
\[
 L(A,\tau)=\sum_{l\in\mathcal F}s_l\log\frac{\kappa_l}{s_l},
 \qquad
 B_{\rm sparse}(A,\mathcal T)
 =\sum_{\substack{r\in\mathcal R\\\kappa_r<b_r/\beta}}b_r.
\]
At $\lambda=1/k$, the same witness gives the explicit refinement
\begin{align*}
 \log\frac1{Q_{1,1/k}(A)}
 \le{}&\Lambda_{1/k}
 +\left(\beta+1+\log3+\log(6eM)\right)|A|\\
 &+\left(2+\log3+\log k\right)|\mathcal F|
   -(1+\log3)|\mathcal R|\\
 &+\frac65L(A,\tau)
   +\frac{6}{5\ln2}\sum_{l\in\mathcal F}\sum_{j\in A_l}F_l^j(A)
   +B_{\rm sparse}(A,\mathcal T).
\end{align*}
Indeed, the size choices cost at most $|A|\log3$.
The child-count choices cost at most
$|\mathcal F|+(|\mathcal F|-|\mathcal R|)\log3$,
and the multiset choices cost at most
$|A|+|\mathcal F|-|\mathcal R|$.
For the subset choices use
\[
 \log Z_{s_l}^l\le s_l\log(6eM),\qquad
 -\sum_{j\in A_l}\log\theta_{s_l}^l(j)
 \le\frac{6}{5\ln2}\sum_{j\in A_l}F_l^j(A)
       +\frac65s_l\log\frac{\kappa_l}{s_l}.
\]
The roots have pairwise disjoint sets $A\cap R_r$, by the source-component
argument above. Their dense rarity costs therefore sum to at most
$\beta|A|$, leaving only $B_{\rm sparse}$.
Finally, root and child gate labels cost at most
$|\mathcal F|\log k$.

Replacing $|\mathcal F|$ by $|A|$, replacing $L(A,\tau)$ by
$|A|\log k$, discarding the negative root term and charging each sparse
root to a coordinate it owns recovers the earlier bound. Retaining these
quantities can be substantially sharper.

\subsection{The addressing benchmark in the same queue construction}

Let $m\ge1$ be the number of address bits and $N=2^m$ the number of
data bits. Write the addressing function as the union of $N$ conjunctions
$h_u$, one for each address $u$, where $h_u$ is true precisely when the
address equals $u$ and the selected data bit is $1$.
Each gate has $d=m+1$ relevant coordinates, true probability $2^{-d}$,
and one regular inner gate. The overlap bound is $k=N$.

Every nonzero Fourier set has the form $A=T\cup\{y_u\}$, with $T$ a
subset of the address coordinates. It can be assigned entirely to
$h_u$. Thus $|\mathcal F|=|\mathcal R|=1$, there are no children,
and $L(A,\tau)=0$.
Put $s=|A|$. Since all ranks are at most $d$, the exact witness bound,
now using the permitted choice $\lambda=1$, yields
\[
 \log\frac1{Q_{1,1}(A)}
 \le\Lambda_1+d+\log(s(s+1))+s\log(6e)
       +\frac65s\log\frac ds+1.
\]
Here
$\Lambda_1=N\log(1+2^{-d})\le(\log e)/2$.
Under the addressing spectral law,
$s=1+\operatorname{Bin}(m,1/2)$, so
$\E s=\I(f)=1+m/2$ and $d+1=2\I(f)$.
Concavity of $s\log(d/s)$ gives
\[
 \mathbb E\left[s\log\frac ds\right]
 \le\I(f)\log\frac d{\I(f)}\le\I(f).
\]
Since $s(s+1)\le3^s$, Gibbs' inequality gives
\[
 \HH(f)\le\frac{\log e}{2}
       +\left(\frac{16}{5}+\log(18e)\right)\I(f)
       <10\,\I(f).
\]
This recovers a constant independent of the number of addresses within
the original queue method. The exact selector calculation gives the
stronger identities $\HH(f)=2m$ and $\I(f)=1+m/2$.

\subsection{Scope of the refinement}

The variational inequality supplies a reference measure and an exact
cost for each finite witness. To obtain full FEI from it, one would still
need a uniform influence bound for the expected joint witness cost,
including its rank, sparse-root and normalization terms.
The addressing estimate proves such a bound for that family.
It does not give one for arbitrary Boolean functions.
The two-cell selector reduction in Section~\ref{sec:selectors} shows that even unrestricted
selector gates already retain the full FEI problem.

\section{Threshold gates under arbitrary product measures}
\label{bias:section}

A symmetric threshold on bits \(y_1,\ldots,y_r\) is
\(T_{t,r}(y)=\mathbf1[\sum_{i=1}^r y_i\ge t]\), \(1\le t\le r\).
We prove product-measure FEI+ for these gates with a constant independent
of the arity and every coordinate bias. This supplies the gate estimate
needed for compositions whose intermediate outputs need not be balanced.

Throughout this section set
\[
 L_0=3+\ln8,\qquad
 A=23+\ln16,\qquad C_T=\frac{20A}{\ln2}<744.
\]
We use natural logarithms for \(\mathcal E_\mu\), and write
\(\HH_\mu^+=\mathcal E_\mu/\ln2\) for its value in bits.

\subsection{Normalization and basic identities}\label{bias:normalization}

Work on \(y_i\in\{0,1\}\), independently with \(0<q_i<1\), and put \[
 \sigma_i=2\sqrt{q_i(1-q_i)},\qquad
 \phi_i(y_i)=\frac{2y_i-1-(2q_i-1)}{\sigma_i}.
\] Then \(\mathbb E\phi_i=0\), \(\mathbb E\phi_i^2=1\), and \(\mathbb E|\phi_i|=\sigma_i\). The products \(\phi_S=\prod_{i\in S}\phi_i\) form an orthonormal basis. For Boolean \(f\), let \[
 P(S)=\widehat f_\mu(S)^2,\quad
 \I_\mu(f)=\sum_S|S|P(S),\quad
 V_\mu(f)=\sum_{S\ne\varnothing}P(S).
\] Define \[
 \mathcal E_\mu(f)=
 \sum_{\varnothing\ne S}P(S)\ln\frac{\sigma_S^2}{P(S)},
 \qquad \sigma_S^2=\prod_{i\in S}\sigma_i^2,\qquad
 \HH_\mu^+(f)=\frac{\mathcal E_\mu(f)}{\ln2}.
\] The quantity \(\Delta_\mu(f)=\I_\mu(f)-V_\mu(f)\) will be called the spectral deficit.

Parseval gives \(\sum_SP(S)=1\). In particular \[
 \Delta_\mu(f)=\sum_{|S|\ge2}(|S|-1)P(S)\ge0.
\] For every \(S\), \[
 |\widehat f_\mu(S)|
   =|\mathbb E f\phi_S|
   \le\mathbb E|\phi_S|=\sigma_S.
\] Consequently each summand in \(\mathcal E_\mu(f)\) is nonnegative.

For \(f=2T-1\), let \(d_i=T(y_i=1)-T(y_i=0)\). Integrating one coordinate gives \[
 \widehat f_\mu(B\cup\{i\})=\sigma_i\widehat d_i(B)
 \quad (i\notin B).
\] For an increasing threshold \(d_i\) is the indicator of a pivotal slice. With \(\rho_i=\mathbb E d_i\), \[
 \I_i=\sum_{S\ni i}P(S)=\sigma_i^2\rho_i,\qquad
 P(\{i\})=\sigma_i^2\rho_i^2.
\] It follows that \[
 J_i:=\sum_{\substack{S\ni i\\|S|\ge2}}P(S)
       =\sigma_i^2\rho_i(1-\rho_i),\qquad
 J=\sum_iJ_i.
\] Since \(k\le2(k-1)\) for \(k\ge2\), \[
 J\le2\Delta_\mu(f).
\]

Under a biased measure, the corresponding bound for uncorrected Shannon entropy fails even for a dictator. For the dictator \(f=2y_1-1\), it is \(h_2(\sigma_1^2)\), while \(\I_\mu=\sigma_1^2\); their ratio diverges as \(q_1\to0\). The corrected entropy \(\HH_\mu^+\) and the deficit \(\Delta_\mu\) are both zero for this dictator.

Deterministic coordinates can be deleted. All boundary cases with \(q_i\in\{0,1\}\) follow by restriction or continuity.

\subsection{Two probability estimates}\label{bias:probability}

\begin{lemma}[Concentration after resampling]\label{bias:concentration}

Let \(N=\sum_i y_i\), with independent Bernoulli probabilities \(q_i\), and \(L=\sum_i\min(q_i,1-q_i)>0\). Then \begin{equation}\label{bias:eq-3-1}
\sup_k\Pr(N=k)\le \frac1{\sqrt{2L}}.
\end{equation}

More generally, fix the original configuration, and refresh coordinate \(i\) independently with probability \(\eta\). A refreshed Bernoulli coordinate can be represented as a fair bit with probability \(2\min(q_i,1-q_i)\), and as its more likely deterministic value otherwise. Conditional on the refresh indicators, the fair-activation indicators, and the values of the nonfair coordinates, the sum is a fixed integer plus a \(\operatorname{Bin}(F,1/2)\) variable, where the independent activation count has mean \[
 m_F=2\eta\sum_i\min(q_i,1-q_i).
\] Uniformly in the original configuration, \begin{equation}\label{bias:eq-3-2}
\sup_k\Pr(N_{\rm refreshed}=k\mid y)
    \le \frac1{\sqrt{m_F}}.
\end{equation}

\end{lemma}
\begin{proof}

By Lemma~\ref{lem:antichain-density}, \[
 \max_k\Pr(\operatorname{Bin}(m,1/2)=k)\le(m+1)^{-1/2}
 \qquad(m\ge0).
\]

If \(F\) is a sum of independent Bernoulli variables with mean \(m_F\), then \[
 \mathbb E\frac1{F+1}
 =\int_0^1\prod_i(1-p_i+p_i s)\,ds
 \le\int_0^1e^{-m_F(1-s)}\,ds\le\frac1{m_F}.
\] Cauchy--Schwarz therefore bounds \(\mathbb E(F+1)^{-1/2}\le m_F^{-1/2}\). Averaging the conditional binomial bound gives \eqref{bias:eq-3-2}. Taking \(\eta=1\), so \(m_F=2L\), gives \eqref{bias:eq-3-1}.

\end{proof}

\begin{lemma}[A small upper tail lies above the mean]\label{bias:tail-above-mean}

Let \(N\) be a sum of independent Bernoulli variables. If \(\Pr(N\ge t)<1/20\), where \(t\) is an integer, then \(t>\mathbb EN\).

\end{lemma}
\begin{proof}

Put \(m=\mathbb EN\), \(v=\operatorname{Var}N\), \(X=N-m\). Independence gives \[
 \mathbb EX^4\le v+3v^2.
\] Hölder and Cauchy--Schwarz give, for \(v>0\), \[
 \mathbb E|X|\ge\frac{v}{\sqrt{1+3v}},\qquad
 \Pr(X>0)\ge
 \frac{(\mathbb EX_+)^2}{\mathbb EX_+^2}
 \ge\frac{v}{4(1+3v)}.
\] If \(t\le m\) and \(v\ge1/2\), this last expression is at least \(1/20\), and \(\{X>0\}\subseteq\{N\ge t\}\). If \(v<1/2\), then \(\{N<t\}\subseteq\{|X|\ge1\}\), so Chebyshev gives \(\Pr(N\ge t)\ge1-v>1/2\). The deterministic case is immediate.

\end{proof}

\subsection{A weak entropy inequality}\label{bias:weak-section}

Let \(f=2T_{t,r}-1\), \(1\le t\le r\), and write \[
 \Lambda=\sum_i\sigma_i^2,\quad
 N=\sum_i y_i,\quad
 \rho_i=\Pr(N-y_i=t-1),\quad \I_i=\sigma_i^2\rho_i,
 \quad \I=\sum_iI_i,\quad V=\operatorname{Var}_\mu f.
\]

\begin{lemma}[A logarithmic spectral-rank estimate]\label{bias:rank}

With \(P(S)=\widehat f_\mu(S)^2\), \begin{equation}\label{bias:eq-4-1}
\sum_{S\ni i}P(S)\ln^+\frac{\Lambda}{|S|}
       \le L_0 \I_i,\qquad
 L_0=3+\ln8.
\end{equation}

\end{lemma}
\begin{proof}

The derivative \(d_i=\mathbf1[N-y_i=t-1]\) has \(\mathbb Ed_i^2=\rho_i\). Its normalized squared coefficients define a probability law \[
 \nu_i(B)=\frac{\widehat d_i(B)^2}{\rho_i}.
\] The left side of \eqref{bias:eq-4-1}, divided by \(\I_i\), is \(\mathbb E_{\nu_i}\ln^+(\Lambda/(|B|+1))\).

If \(\Lambda\le2\) this is at most \(\ln2\). Suppose \(\Lambda>2\). Refresh the other coordinates independently with probability \(\eta\). The orthonormal noise identity gives \[
 \mathbb E[d_i(y)d_i(z)]
 =\rho_i\,\mathbb E_{\nu_i}(1-\eta)^{|B|}.
\] For the conditional sum of the refreshed other coordinates, the free-bit count in Lemma~\ref{bias:concentration} has mean at least \[
 2\eta\sum_{j\ne i}\min(q_j,1-q_j)
 \ge \frac\eta2(\Lambda-\sigma_i^2)
 \ge\frac{\eta\Lambda}{4}.
\] It follows that \[
 \mathbb E_{\nu_i}(1-\eta)^{|B|}
 \le\frac2{\sqrt{\eta\Lambda}}.
\] For real \(K>1\), take \(\eta=1/(2K)\). The event relevant to the logarithmic moment is \(\{|B|+1<K\}\), so \[
 \begin{aligned}
 \nu_i(|B|+1<K)
 &\le(1-\eta)^{-(K-1)}
          \mathbb E_{\nu_i}(1-\eta)^{|B|}\\
 &\le2\sqrt{2e}\sqrt{\frac K\Lambda}.
 \end{aligned}
\] Indeed \[
 -(K-1)\ln\left(1-\frac1{2K}\right)
 \le\frac{K-1}{2K-1}<\frac12,
\] using \(-\ln(1-z)\le z/(1-z)\). For \(K\le1\), the event is empty. Therefore \[
 \mathbb E_{\nu_i}\ln^+\frac{\Lambda}{|B|+1}
 \le\int_0^\infty
       \min\left(1,2\sqrt{2e}\,e^{-u/2}\right)du
 =2\ln(2\sqrt{2e})+2=3+\ln8.
\]

\end{proof}

\begin{lemma}[The normalization credit]\label{bias:credit}

\begin{equation}\label{bias:eq-4-2}
V\ln\frac1V
 \le\sum_i \I_i\ln^+\frac{4\rho_i}{V}
       +(1+\ln2+13)\I.
\end{equation}

\end{lemma}
\begin{proof}

Let \(\alpha=\min(\Pr(f=1),\Pr(f=-1))\). If \(\alpha\ge1/20\), then \(V=4\alpha(1-\alpha)\ge19/100\), and \[
 V\ln(1/V)\le\ln(100/19)V\le\ln(100/19)\I,
\] which is more than sufficient. Otherwise negate all inputs and the output if needed so that the minority event is \(A=\{N\ge t\}\), of probability \(\alpha<1/20\). This transformation preserves all the displayed quantities. By Lemma~\ref{bias:tail-above-mean}, \(t>\mathbb EN\).

Write \(\mu_A=\mu(\cdot\mid A)\). The relative-entropy decomposition against a product measure is \begin{equation}\label{bias:eq-4-3}
\ln\frac1\alpha
 =\mathrm{TC}(\mu_A)+
     \sum_i D_{\rm KL}(\mu_A^{(i)}\Vert\mu_i),
\end{equation} where \(\mathrm{TC}(\mu_A)\) is the relative entropy of \(\mu_A\) with respect to the product of its marginals. This identity follows directly by splitting the logarithm of the density.

Fix \(i\), let \(a=\Pr(A\mid y_i=1)\), \(b=\Pr(A\mid y_i=0)\), so \(a-b=\rho_i\) and \(\alpha=q_i a+(1-q_i)b\). Since \(b\le\alpha\), and with \(u=(1-q_i)\rho_i/\alpha\), \begin{equation}\label{bias:eq-4-4}
\begin{aligned}
 \alpha D_{\rm KL}(\mu_A^{(i)}\Vert\mu_i)
 &=q_i a\ln(a/\alpha)+(1-q_i)b\ln(b/\alpha)\\
 &\le q_i\alpha(1+u)\ln(1+u)\\
 &\le q_i(1-q_i)\rho_i[1+\ln(1+u)]\\
 &\le\frac{\I_i}{4}
       \left[1+\ln2+\ln^+\frac{4\rho_i}{V}\right].
 \end{aligned}
\end{equation} Here \((1+u)\ln(1+u)\le u[1+\ln(1+u)]\), \(\ln(1+u)\le\ln2+\ln^+u\), and \(\alpha\ge V/4\).

It remains to establish \begin{equation}\label{bias:eq-4-5}
\alpha\,\mathrm{TC}(\mu_A)\le\frac{13}{4}\I.
\end{equation} If \(t=r\), conditioning on \(A\) fixes every input and the total correlation is zero. Otherwise exponentially tilt the product measure: \[
 \frac{dR}{d\mu}=\frac{e^{\lambda N}}M,
 \qquad \lambda>0,\qquad \mathbb E_RN=t.
\] Such a finite \(\lambda\) exists because \(\mathbb EN<t<r\). Put \(b_*=\Pr_R(N=t)\). For every product measure \(R\), \(D_{\rm KL}(\mu_A\Vert R)\ge\mathrm{TC}(\mu_A)\), and here \begin{equation}\label{bias:eq-4-6}
D_{\rm KL}(\mu_A\Vert R)
   =\ln(M/\alpha)-\lambda\mathbb E_{\mu_A}N
   \le\ln(Me^{-\lambda t}/\alpha)
   \le\ln(1/b_*).
\end{equation} The last step uses \(\alpha=M\mathbb E_R[e^{-\lambda N}\mathbf1_A]  \ge Me^{-\lambda t}b_*\).

Let \(r_i=R(y_i=1)\), \(a_i=1-q_i+q_i e^\lambda\), so \(a_i\ge1\). The complementary upper bound \(\alpha\le Me^{-\lambda t}\), together with the slice change of measure, gives \[
 \frac{\I_i}{\alpha}
 \ge\frac{\sigma_i^2e^\lambda}{a_i}
       R(N-y_i=t-1)
 =4r_i(1-r_i)a_iR(N-y_i=t-1)
 \ge4r_i(1-r_i)R(N-y_i=t-1).
\] The identity \[
 \mathbb E_R[(y_i-r_i)h(N)]
   =r_i(1-r_i)\mathbb E_R[h(N-y_i+1)-h(N-y_i)]
\] holds by conditioning on the other coordinates. Sum it for \(h(N)=\mathbf1[N\ge t]\). As \(\mathbb E_RN=t\), this yields \[
 \frac \I\alpha\ge4\mathbb E_R(N-t)_+
                  =2\mathbb E_R|N-t|.
\] Writing \(v=\operatorname{Var}_RN\), the fourth-moment estimate and Hölder used in Lemma~\ref{bias:tail-above-mean} give \begin{equation}\label{bias:eq-4-7}
\frac \I\alpha\ge\frac{2v}{\sqrt{1+3v}}.
\end{equation}

We use the following classical theorem of Darroch \cite{Darroch64} (see also \cite{Pitman97}): if a finite sum of independent Bernoulli variables has integer mean \(t\), then its mass at \(t\) is maximal among its point probabilities.

If \(v\le1/2\), the integer support and Chebyshev give \(b_*\ge1-v\), hence \(\mathrm{TC}(\mu_A)\le-\ln(1-v)\le2v\ln2\). By \eqref{bias:eq-4-7}, \[
 \frac{\alpha\,\mathrm{TC}(\mu_A)}\I
 \le\ln2\sqrt{1+3v}<1.1.
\] If \(v>1/2\), let \(K=\lceil2\sqrt v\rceil\). Chebyshev and the mode theorem give \[
 b_*\ge\frac{3}{4(2K+1)}
       \ge\frac3{4(4\sqrt v+3)}.
\] Consequently \[
 \frac{\alpha\,\mathrm{TC}(\mu_A)}\I
 \le\psi(v):=
 \ln\frac{16\sqrt v+12}{3}\,
      \frac{\sqrt{1+3v}}{2v}.
\] This function decreases on \([1/2,\infty)\) and \(\psi(1/2)<3.25\). For completeness, put \(x=\sqrt v\). The logarithmic derivative of \(\sqrt{1+3x^2}/(2x^2)\) is strictly less than \(-1/x\). The logarithmic derivative of \(\ln((16x+12)/3)\) is strictly less than \(1/x\) for \(x\ge1/\sqrt2\), because its logarithm exceeds 1. This proves the decrease. For the endpoint bound, use \(\sqrt{5/2}<19/12\), \(\sqrt2<99/70\), and \[
 e^{39/19}>
 \left(\sum_{k=0}^8\frac{2^k}{k!}\right)\left(1+\frac1{19}\right)
 =\frac{2327}{315}\frac{20}{19}
 >\frac{272}{35}>
       \frac{8\sqrt2+12}{3}.
\] Thus \(\psi(1/2)<(19/12)(39/19)=13/4\). Thus \eqref{bias:eq-4-5} holds.

Finally \(V\ln(1/V)\le4\alpha\ln(1/\alpha)\). Combining \eqref{bias:eq-4-3}, \eqref{bias:eq-4-4}, and \eqref{bias:eq-4-5} proves \eqref{bias:eq-4-2}.

\end{proof}

\begin{proposition}[The weak bound]\label{bias:weak}

\begin{equation}\label{bias:eq-4-8}
\mathcal E_\mu(f)\le A \I,\qquad
 A=6+L_0+(1+\ln2+13)=23+\ln16<25.8.
\end{equation}

\end{proposition}
\begin{proof}

Let \(W_k=\sum_{|S|=k}P(S)\), \(m=\I/V\), and \(\pi_i=\I_i/V\). For \(k\ge1\), define positive reference weights \[
 \theta_k(i)=
 \max\left(\sigma_i^2\min(1,k/\Lambda),\,4\pi_i\right).
\] Let \(Z_k\) be their \(k\)-th elementary symmetric polynomial. Gibbs' inequality on the sets of size \(k\) gives \begin{equation}\label{bias:eq-4-9}
\sum_{|S|=k}P(S)\ln\frac{\sigma_S^2}{P(S)}
 \le W_k\ln\frac1{W_k}+W_k\ln Z_k
       +\sum_{|S|=k}P(S)\sum_{i\in S}
           \ln\frac{\sigma_i^2}{\theta_k(i)}.
\end{equation} Since \(\sum_i\theta_k(i)\le k+4m\), \[
 Z_k\le\left(\frac{e(k+4m)}k\right)^k,\qquad
 \ln Z_k\le k+4m.
\] Moreover \[
 \ln\frac{\sigma_i^2}{\theta_k(i)}
 \le\ln^+\frac\Lambda k-\ln^+\frac{4\rho_i}{V}.
\] The level distribution \(W_k/V\), supported on positive integers with mean \(m\), has entropy at most \(1+\ln m\), by Gibbs' inequality against the geometric distribution with that mean. Hence \[
 \sum_kW_k\ln\frac1{W_k}
 \le V\ln\frac1V+V(1+\ln m)
 \le V\ln\frac1V+\I.
\] Summing \eqref{bias:eq-4-9}, and then applying Lemmas~\ref{bias:rank} and~\ref{bias:credit}, gives \[
 \begin{aligned}
 \mathcal E_\mu(f)
 &\le V\ln\frac1V+6I
      +\sum_i\sum_{S\ni i}P(S)\ln^+\frac\Lambda{|S|}
      -\sum_iI_i\ln^+\frac{4\rho_i}{V}\\
 &\le[6+L_0+(1+\ln2+13)]\I.
 \end{aligned}
\]

\end{proof}

\subsection{The strengthened inequality}\label{bias:strong-section}

\begin{theorem}[Threshold gates under product measures]\label{bias:strong}

For every \(T_{t,r}\) and every product law, \begin{equation}\label{bias:eq-5-1}
\mathcal E_\mu(2T_{t,r}-1)\le10A J\le20A \Delta_\mu(2T_{t,r}-1).
\end{equation} In particular \(\HH_\mu^+\le C_T \Delta_\mu\), where \(C_T=20A/\ln2<744\). For the displayed constant, \(C_T=80+460/\ln2\), and the convergent series \[
 \ln2=2\sum_{k\ge0}\frac1{(2k+1)3^{2k+1}}
   >\frac{842}{1215}>\frac{115}{166}
\] gives \(C_T<744\).

\end{theorem}
\begin{proof}

The nonconstant singleton coefficients contribute \begin{equation}\label{bias:eq-5-2}
\sum_i2\sigma_i^2\rho_i^2\ln(1/\rho_i)
 \le2\sum_i\sigma_i^2\rho_i(1-\rho_i)=2J.
\end{equation} If every \(\rho_i\le9/10\), then \(J\ge \I/10\) and Proposition~\ref{bias:weak} proves \eqref{bias:eq-5-1}.

Suppose instead that a largest pivotal probability is \(\rho_0=1-\varepsilon>9/10\). The case \(\varepsilon=0\) is a dictator after deleting deterministic coordinates, so has zero entropy and deficit. Assume \(0<\varepsilon<1/10\). Negating all inputs and the output if needed, arrange \(q=q_0\le1/2\).

For the other coordinates let \(\bar y_i\) be the more likely input value and \(z_i=\mathbf1[y_i\ne\bar y_i]\). Let \(U\) consist of the coordinates whose more likely value is 0 and \(D\) those whose more likely value is 1. Put \[
 L_U=\sum_{i\in U}\mathbb Ez_i,\quad
 L_D=\sum_{i\in D}\mathbb Ez_i,\quad L=L_U+L_D,
 \quad
 \Lambda_U=\sum_{i\in U}\sigma_i^2,\quad
 \Lambda_D=\sum_{i\in D}\sigma_i^2,\quad
 \Lambda'=\Lambda_U+\Lambda_D.
\] We first establish \begin{equation}\label{bias:eq-5-3}
L<\frac18,\qquad L\le\frac87\varepsilon,\qquad
 \Lambda'<\frac12,\qquad
 \Lambda'\le\frac{32}{7}\varepsilon,\qquad
 2L_U\le\Lambda_U,\quad2L_D\le\Lambda_D.
\end{equation}

\textbf{Step 1: proof of \eqref{bias:eq-5-3}.} Lemma~\ref{bias:concentration} applied to \(N'=N-y_0\) gives \(L\le1/(2\rho_0^2)<4/5\). If the most likely configuration had \(|D|\ne t-1\), the no-error event would fail the pivotal slice, so \[
 \prod_{i\ne0}(1-\mathbb Ez_i)\le\varepsilon<1/10.
\] But \(\ln(1-u)\ge-3u/2\) on \(0\le u\le1/2\); hence this would imply \(L>(2/3)\ln10>4/3\), a contradiction. Thus \(|D|=t-1\), and \(N'-(t-1)=Z_U-Z_D\).

An outcome with exactly one error lies off the pivotal slice. By independence and the union bound on all other errors, \[
 \varepsilon\ge\Pr\!\left(\sum z_i=1\right)\ge L(1-L).
\] On \([1/8,4/5]\), the minimum of \(L(1-L)\) exceeds \(1/10\). Thus \(L<1/8\), and then \(\varepsilon\ge(7/8)L\). Finally \(2u\le4u(1-u)\le4u\) for \(u\in[0,1/2]\), proving the remaining bounds in \eqref{bias:eq-5-3}.

\textbf{Step 2: lower bounds for the available deficit.} Write \(J_U=\sum_{i\in U}J_i\), \(J_D=\sum_{i\in D}J_i\), and \(J_0=\sigma_0^2\varepsilon(1-\varepsilon)\). For \(i\in U\), the no-other-error configurations with \(y_0=0\) show \(\rho_i\ge(1-q)(1-L)>7/16\), while the choice of coordinate 0 gives \(\rho_i\le1-\varepsilon\). Concavity of \(x(1-x)\) therefore yields \begin{equation}\label{bias:eq-5-4}
J_U\ge\frac9{10}\varepsilon\Lambda_U.
\end{equation} For \(i\in D\), put \(a=1-L>7/8\). The no-other-error configurations with \(y_0=1\) or \(y_0=0\) give \[
 aq\le\rho_i\le1-a(1-q).
\] The smaller value of \(x(1-x)\) at these endpoints is \(aq(1-aq)\), because their difference is \(a(1-a)(1-2q)\ge0\). Consequently \begin{equation}\label{bias:eq-5-5}
\rho_i(1-\rho_i)\ge
       \frac78q\left(1-\frac78q\right)
       \ge\frac{63}{128}q,
 \qquad
 J_D\ge\frac{63}{128}q\Lambda_D.
\end{equation}

\textbf{Step~3: coefficients containing coordinate 0.} Let \(d=\mathbf1[N'=t-1]\), \(v_d=\operatorname{Var}d=\varepsilon(1-\varepsilon)\). The coefficients indexed by \(\{0\}\cup B\), \(B\ne\varnothing\), contribute \begin{equation}\label{bias:eq-5-6}
\sigma_0^2
    \sum_{B\ne\varnothing}\widehat d(B)^2
              \ln\frac{\sigma_B^2}{\widehat d(B)^2}.
\end{equation} Let \(U_k=\sum_{|B|=k}\widehat d(B)^2\), and let \(e_k\) be the \(k\)-th elementary symmetric polynomial of the \(\sigma_i^2\) for \(i\ne0\). By \eqref{bias:eq-5-3}, \[
 e_k\le(\Lambda')^k\le(6v_d)^k,
 \quad \sum_{k\ge1}U_k=v_d.
\] Levelwise Gibbs and \(x\ln(a/x)\le a/e\) give \[
 \begin{aligned}
 \sum_{B\ne\varnothing}\widehat d(B)^2
              \ln\frac{\sigma_B^2}{\widehat d(B)^2}
 &\le\sum_k U_k\ln\frac{e_k}{U_k}\\
 &\le\sum_k\frac{v_d^k}{e}
             +(\ln6)\sum_k kU_k\\
 &\le\frac{10}{9e}v_d+(\ln6)\I_\mu(d).
 \end{aligned}
\] Here \(\I_\mu(d)\) denotes the spectral Dirichlet energy of the \(0/1\)-valued function \(d\). Its difference in any coordinate has absolute value at most 1, so \[
 \I_\mu(d)\le\sum_{i\ne0}q_i(1-q_i)
       =\Lambda'/4\le\frac{80}{63}v_d.
\] The resulting coefficient is below 3: since \(e>8/3\) and \(\ln6<2\), it is less than \(5/12+160/63=745/252<3\). Thus \eqref{bias:eq-5-6} is at most \(3J_0\).

\textbf{Step~4: coefficients avoiding coordinate 0.} Put \[
 u=\mathbf1[Z_U-Z_D\ge1],\qquad
 v=\mathbf1[Z_D-Z_U\ge1].
\] Then \(T=u+y_0(1-u-v)\). For nonempty \(B\subseteq U\cup D\), the coefficient of the \(\pm1\)-valued function is \[
 2(1-q)\widehat u(B)-2q\widehat v(B).
\] Every \(0/1\)-valued function \(w\) satisfies \begin{equation}\label{bias:eq-5-7}
|\widehat w(B)|\le\sigma_B/2\quad(B\ne\varnothing).
\end{equation} Indeed integrate a coordinate \(i\in B\): \(\widehat w(B)=(\sigma_i/2)\mathbb E[\Delta_iw\,\phi_{B\setminus\{i\}}]\), and use \(|\Delta_iw|\le1\). Moreover \begin{equation}\label{bias:eq-5-8}
|\widehat u(B)|\le L_U\sigma_B\quad(B\subseteq D),\qquad
 |\widehat v(B)|\le L_D\sigma_B\quad(B\subseteq U),
\end{equation} because \(0\le\mathbb E[u\mid y_D]\le\Pr(Z_U\ge1)\le L_U\), and likewise for \(v\).

The following elementary inequality will handle the difference of coefficients. For \(s>0\), define \(\varphi_s(x)=x\ln(s/x)\). If \(a^2,b^2\le s\), then \begin{equation}\label{bias:eq-5-9}
\varphi_s((a-b)^2)
 \le2\varphi_s(a^2)+2\varphi_s(b^2)+2a^2+2b^2.
\end{equation} If \(2a^2+2b^2\le s/e\), use monotonicity of \(\varphi_s\) up to \(s/e\), followed by the subadditivity \(\varphi_s(x+y)\le\varphi_s(x)+\varphi_s(y)\). If \(2a^2+2b^2>s/e\), simply use \(\varphi_s((a-b)^2)\le s/e<2a^2+2b^2\). These are the two cases of \eqref{bias:eq-5-9}.

The function \(x[1+\ln(s/x)]\) increases on \([0,s]\). Therefore, on a family of coefficients satisfying \(a_B^2\le\lambda\sigma_B^2\), \(0\le\lambda\le1\), the corresponding terms on the right of \eqref{bias:eq-5-9} sum to at most \begin{equation}\label{bias:eq-5-10}
2\lambda[1+\ln(1/\lambda)]
       \sum_B\sigma_B^2.
\end{equation} Zero parameters are interpreted by continuity.

We also have the product-sum estimates \begin{equation}\label{bias:eq-5-11}
\sum_{\substack{|B|\ge2\\B\cap U\ne\varnothing}}\sigma_B^2
 \le\Lambda_U(e^{\Lambda'}-1)
 \le\frac53\Lambda'\Lambda_U,
\end{equation} \begin{equation}\label{bias:eq-5-12}
\sum_{\substack{B\subseteq D\\|B|\ge2}}\sigma_B^2
 \le\frac12 e^{\Lambda_D}\Lambda_D^2
 \le\frac56\Lambda_D^2,
\end{equation} and the analogous bounds with \(U,D\) interchanged. Here \(e^{1/2}<5/3\) and \eqref{bias:eq-5-3} were used. The scalar bounds used below follow, for example, from \(8/3<e<11/4\) and \(\ln2<7/10\), obtained directly from the exponential series. The first estimate follows by choosing one member of \(B\cap U\), possibly overcounting; the second follows by the exponential series.

Let \(a_B=2(1-q)\widehat u(B)\) and \(b_B=2q\widehat v(B)\). Apply \eqref{bias:eq-5-10} to the following four families.

\begin{itemize}
\item
  \textbf{\(a_B\), sets meeting \(U\):} use \(\lambda=1\). By \eqref{bias:eq-5-11} and \eqref{bias:eq-5-3}, their contribution is at most \[
    \frac{10}{3}\Lambda'\Lambda_U
    \le\frac{320}{21}\varepsilon\Lambda_U
    <16\varepsilon\Lambda_U.
  \]
\item
  \textbf{\(a_B\), sets contained in \(D\):} put \(x=2L_U<1/4\) and use \(\lambda=x^2\). The elementary bound \(x(1+2\ln(1/x))\le(1+2\ln4)/4<1\) on this interval gives a contribution at most \[
   \frac53 x^2(1+2\ln(1/x))\Lambda_D^2
   \le\frac53\Lambda_U(\Lambda')^2
   \le\frac{80}{21}\varepsilon\Lambda_U
   <4\varepsilon\Lambda_U.
  \]
\item
  \textbf{\(b_B\), sets meeting \(D\):} use \(\lambda=q^2\) and \(q(1+2\ln(1/q))\le(1+2\ln2)/2<6/5\). Their contribution is at most \[
   \frac{10}{3}q^2(1+2\ln(1/q))\Lambda'\Lambda_D
   \le2q\Lambda_D.
  \]
\item
  \textbf{\(b_B\), sets contained in \(U\):} put \(y=2L_D<1/4\) and use \(\lambda=q^2y^2\). Since \[
    qy(1+2\ln(1/q))+2qy\ln(1/y)
       \le\frac3{10}+\frac{\ln4}{4}<\frac23,
  \] their contribution is at most \[
   \frac53 q^2y^2[1+2\ln(1/q)+2\ln(1/y)]\Lambda_U^2
   \le\frac{10}{9}qy\Lambda_U^2
   \le\frac5{18}q\Lambda_D
   <\frac13q\Lambda_D.
  \]
\end{itemize}

Thus all degree-at-least-two coefficients avoiding 0 contribute at most \begin{equation}\label{bias:eq-5-13}
20\varepsilon\Lambda_U+\frac73q\Lambda_D
 \le23J_U+5J_D,
\end{equation} by \eqref{bias:eq-5-4}--\eqref{bias:eq-5-5}.

Adding \eqref{bias:eq-5-2}, Step~3, and \eqref{bias:eq-5-13} gives \[
 \mathcal E_\mu(f)\le2J+3J_0+23J_U+5J_D\le25J.
\] Since \(25<10A\), this proves the remaining case of \eqref{bias:eq-5-1}.

\end{proof}

\section{Composition and strong entropy}\label{comp:section}

We use the product-measure composition principle of O'Donnell and Tan \cite{ODT13}. The exact identities below record the entropy and deficit cancellation.

Let \(f_i\) be Boolean functions on mutually disjoint input blocks with product laws \(\mu_i\). Let \(\eta_i=\mathbb E_{\mu_i}f_i\) and \(v_i=1-\eta_i^2>0\). Let \(g\) be a Boolean function on the independent outputs, endowed with the product law \(\eta\) in \(\{-1,1\}\) notation. Set \(F=g(f_1,\ldots,f_m)\). Let \[
 a_i=\frac{\I_i^\eta(g)}{v_i}.
\]

\begin{proposition}[Exact composition identities]\label{comp:identities}

\begin{equation}\label{comp:eq-6-1}
\mathcal E_\mu(F)=\mathcal E_\eta(g)
                   +\sum_i a_i\mathcal E_{\mu_i}(f_i),
\end{equation} \begin{equation}\label{comp:eq-6-2}
\I_\mu(F)=\sum_i a_i \I_{\mu_i}(f_i),
 \qquad
 \operatorname{Var}_\mu(F)=\operatorname{Var}_\eta(g),
\end{equation} \begin{equation}\label{comp:eq-6-3}
\Delta_\mu(F)=\Delta_\eta(g)+\sum_i a_i\Delta_{\mu_i}(f_i).
\end{equation}

\end{proposition}
\begin{proof}

The centered normalized output of block \(i\) expands as \[
 \frac{f_i-\eta_i}{\sqrt{v_i}}
    =\sum_{T\ne\varnothing}
           \frac{\widehat f_i(T)}{\sqrt{v_i}}\phi_T.
\] For a nonempty union \(S=\bigcup_{i\in A}T_i\), where every \(T_i\) is nonempty and contained in block \(i\), substitution into the \(\eta\)-biased Fourier expansion of \(g\) gives \begin{equation}\label{comp:eq-6-4}
\widehat F_\mu(S)
    =\widehat g_\eta(A)
        \prod_{i\in A}\frac{\widehat f_i(T_i)}{\sqrt{v_i}}.
\end{equation} There is exactly one possible \(A\) and one collection \((T_i)\) for each \(S\), so there is no interference between different decompositions. Equivalently, sample \(A\) from the squared Fourier spectrum of \(g\), and for each \(i\in A\) independently sample a nonempty \(T_i\) from the law \(\widehat f_i(T_i)^2/v_i\). The union has the squared Fourier law of \(F\).

The logarithm in the corrected entropy splits: \[
 \ln\frac{\sigma_S^2}{\widehat F_\mu(S)^2}
 =
 \ln\frac{\prod_{i\in A}v_i}{\widehat g_\eta(A)^2}
  +\sum_{i\in A}
       \ln\frac{\sigma_{T_i}^2}{\widehat f_i(T_i)^2}.
\] Taking expectations gives \eqref{comp:eq-6-1}, because the probability that \(i\in A\) is \(\I_i^\eta(g)\). Taking expected cardinalities in \eqref{comp:eq-6-4} gives the influence identity. The variance identity follows from the output distribution, or from its constant coefficient. Finally \[
 \sum_i a_iv_i=\I_\eta(g),
\] which proves \eqref{comp:eq-6-3}.

\end{proof}

\begin{corollary}[Constant-preserving composition]\label{comp:closure}

If \(g\) satisfies biased FEI+ with bit constant \(C_0\) under its induced product law and every \(f_i\) satisfies biased FEI+ with bit constant \(C_i\), then \[
 \HH_\mu^+(F)\le\max(C_0,C_1,\ldots,C_m)\Delta_\mu(F).
\]

\end{corollary}
\begin{proof}

This is immediate from \eqref{comp:eq-6-1} and \eqref{comp:eq-6-3}, since every deficit and \(a_i\) is nonnegative. Constant child outputs are fixed and removed before applying the identity.

\end{proof}

\begin{corollary}[Read-once threshold and parity formulas]\label{comp:threshold-formulas}

Every read-once formula built from \(T_{t,r}\) gates and parity gates, with arbitrary fan-ins and negations, satisfies \[
 \HH_\mu^+(f)\le C_T\Delta_\mu(f),\qquad C_T<744,
\] under every product measure on its original inputs. Depth and intermediate biases do not appear in the constant.

\end{corollary}
\begin{proof}

Literals have zero corrected entropy and deficit. Threshold gates satisfy Theorem~\ref{bias:strong}; parity gates satisfy the smaller constant in Lemma~\ref{comp:elementary-gates}. Apply Corollary~\ref{comp:closure} from the leaves to the root.

\end{proof}

\begin{corollary}[Hybrid formulas with LTF blocks]\label{comp:hybrid}

By Theorem~\ref{thm:A}, every linear threshold function satisfies \(\HH(h)\le C_L\Delta(h)\), with \(C_L=4\Creg\). Take disjoint uniform-input LTF terminal blocks and combine them in a read-once tree in which:

\begin{itemize}
\item
  symmetric threshold and parity gates accept inputs of any means;
\item
  an arbitrary weighted LTF gate may accept balanced nonconstant child outputs;
\item
  fixed-arity arbitrary Boolean gates may be included with their established product-biased FEI+ constants.
\end{itemize}

Then the complete function has \[
 \HH^+(F)\le C \Delta(F),
\] where \(C\) is the maximum of the applicable constants, including \(C_L\) and \(C_T\). Consequently \(\HH(F)\le C\I(F)\), since \(C\ge C_L>\log_2e\), and \(\HH(F)\le(C+\kappa)\Delta(F)\).

\end{corollary}
\begin{proof}

Apply Corollary~\ref{comp:closure} at each gate. The incoming product law at every weighted LTF gate is uniform by the balance hypothesis. The full entropy conclusions follow from Corollary~\ref{comp:full-entropy}.

\end{proof}

\subsection{AND, OR, and parity}\label{comp:elementary-section}

The following elementary gate bounds will also be useful.

\begin{lemma}[Elementary gate constants]\label{comp:elementary-gates}

Under every product measure, \begin{equation}\label{comp:eq-7-1}
\HH_\mu^+(\mathrm{AND}_r),\,\HH_\mu^+(\mathrm{OR}_r)
       \le\left(2+\frac4{\ln2}\right)\Delta_\mu,
 \qquad
 \HH_\mu^+(\mathrm{XOR}_r)\le\frac2{\ln2}\Delta_\mu.
\end{equation} The AND/OR constant \(2+4/\ln2=7.770780\ldots\) is sharp over product laws and arities.

\end{lemma}
\begin{proof}

For \(a_i\in(0,1]\), put \[
 Z=\prod_i a_i,\quad
 R=\sum_i(1-a_i)-1+Z,\quad
 \Phi=\sum_i(a_i-Z)\ln(1/a_i).
\] The scalar inequality \begin{equation}\label{comp:eq-7-2}
\Phi\le2R
\end{equation} follows by induction. Appending \(b\in(0,1]\) changes these quantities by \[
 \Delta R=(1-b)(1-Z),
\] \[
 \Delta\Phi=(1-b)Z\ln(1/Z)+b\ln(1/b)(1-Z)
            \le2(1-b)(1-Z),
\] using \(s\ln(1/s)\le1-s\).

For an OR gate take \(a_i=1-q_i\), so the complementary \(0/1\) indicator is \(G=\prod_i(1-y_i)\). Its \(\pm1\) Fourier probabilities on nonempty sets are \[
 P(S)=4\prod_{i\in S}\frac{\sigma_i^2}{4}
              \prod_{i\notin S}a_i^2.
\] Direct multiplication gives \[
 \I=4Z\sum_i(1-a_i),\qquad V=4Z(1-Z),\qquad \Delta=4ZR,
\] and \[
 \mathcal E_\mu=(2\ln2)\Delta+8Z\Phi
        \le(2\ln2+4)\Delta.
\] Input and output negations handle AND. For OR on two bits with \(q_1=q_2=q\to0\), \(R=q^2\) and \(\Phi=2q(1-q)\ln(1/(1-q))\sim2q^2\). The bit ratio tends to \(2+4/\ln2\), proving sharpness.

For parity write \(\nu_i=2q_i-1\) and \(a_i=\nu_i^2\). Its spectrum is \[
 P(S)=\sigma_S^2\prod_{i\notin S}a_i.
\] The same scalar quantities satisfy \(\Delta=R\), \(\mathcal E_\mu=\Phi\), and \eqref{comp:eq-7-2} proves the parity estimate. Zero \(a_i\) follow by continuity.

\end{proof}

\subsection{From FEI+ to full strong entropy}\label{comp:conversion-section}

\begin{lemma}[The empty coefficient and the spectral deficit]\label{comp:empty-atom}

For every uniform-input Boolean \(f\), \begin{equation}\label{comp:eq-8-1}
h_0(f):=(\mathbb Ef)^2\log_2\frac1{(\mathbb Ef)^2}
       \le\kappa \Delta(f),
 \qquad \kappa=\sqrt{\frac2{1-\ln2}}<2.553.
\end{equation}

\end{lemma}
\begin{proof}

Let \(p=\min(\Pr(f=1),\Pr(f=-1))\), \(x=1-2p=|\mathbb Ef|\). The edge-isoperimetric bound is \begin{equation}\label{comp:eq-8-2}
\I(f)\ge2p\log_2(1/p).
\end{equation} One direct proof of \eqref{comp:eq-8-2} uses entropy tensorization for the minority indicator \(u\): \(\operatorname{Ent}(u)\le\sum_i\mathbb E\operatorname{Ent}_i(u)\). Its left side is \(p\ln(1/p)\); each pivotal fiber contributes \((\ln2)/2\), giving the right side \((\ln2)\I(f)/2\). Entropy tensorization itself follows by induction from the conditional entropy identity and convexity of the entropy functional; it is the standard finite product entropy inequality.

As \(V=1-x^2\), \eqref{comp:eq-8-2} gives \[
 \Delta(f)\ge
 \frac{1-x}{\ln2}
      [-\ln(1-x)-x\ln2].
\] For \(0<x<1\), put \(a=1-\ln2>0\). Then \[
 -\ln(1-x)-x\ln2\ge ax+x^2/2.
\] Also \[
 -x\ln x\le\sqrt x(1-x).
\] To check the latter, set \(x=e^{-2t}\): it becomes \(t\le\sinh t\). Therefore \[
 \frac{h_0(f)}{\Delta(f)}
 \le\frac{-2x\ln x}{(1-x)(a+x/2)}
 \le\frac{2\sqrt x}{a+x/2}
 \le\sqrt{\frac2a}=\kappa,
\] where the last step is the arithmetic-geometric mean inequality. The endpoints \(x=0,1\) have \(h_0=0\), and follow directly or by continuity.

\end{proof}

\begin{corollary}[Full strong entropy]\label{comp:full-entropy}

If \(\HH^+(f)\le C \Delta(f)\) on the uniform cube, then \[
 \HH(f)\le(C+\kappa)\Delta(f).
\] Also \[
 \HH(f)\le\max(C,\log_2e)\,\I(f),
\] since \(h_0(f)\le(\log_2e)V(f)\).

\end{corollary}
\begin{proof}

The first estimate follows from Lemma~\ref{comp:empty-atom}. The second follows from \(\HH=\HH^++h_0\), \(h_0\le(\log_2e)V\), and \(\Delta+V=\I\).

\end{proof}

Apply these facts with \(C=C_T<744\) to obtain the stated clean formula bounds \(\HH(f)\le744I(f)\) and \(\HH(f)\le747\Delta(f)\).

\subsection{Balanced regular formulas and mixed gate classes}

The uniform-input results also give a smaller constant when every nonconstant output below the root is balanced.

\begin{corollary}\label{comp:balanced-regular}
Let $f$ be a read-once formula whose gates are regular Boolean functions or character multiples of regular Boolean functions. Suppose every nonconstant gate output below the root is balanced under the uniform measure on the original inputs. Then
\[
 \HH(f)\le4\Creg\bigl(\I(f)-\Var(f)\bigr).
\]
In particular this applies to formulas of arbitrary depth and fan-in over balanced weighted threshold gates, with an arbitrary weighted threshold at the root.
\end{corollary}
\begin{proof}
Remove constant children by restriction; the gate classes are closed under this operation. If the nonconstant children $f_i$ are balanced, then $v_i=1$, and the spectral substitution in \eqref{comp:eq-6-4} gives the full entropy identity
\[
 \HH(g(f_1,\ldots,f_m))
   =\HH(g)+\sum_i\Inf_i(g)\HH(f_i).
\]
Combine it with \eqref{comp:eq-6-3} and Theorem~\ref{thm:A}, and proceed from the leaves to the root.
\end{proof}

The same argument allows the other uniform gate classes proved above. More generally, consider a read-once tree with uniform original inputs. At a gate whose nonconstant children are balanced, one may use any of the following classes with the indicated FEI+ constant:
\[
\begin{array}{ll}
 \text{character multiples of regular functions:}
    &4\Creg,\\[2pt]
 \text{character multiples of the combinations in Theorem~\ref{thm:B}:}
    &4\Cmulti,\\[2pt]
 \text{unions or intersections as in Theorem~\ref{thm:readk-intro}:}
    &132+\dfrac{118}{5}\log k+\dfrac{106}{5}\log M.
\end{array}
\]
For a chain width $w$, the first two constants increase by $4\log w$ and the last by $2\log w$. Symmetric threshold and parity gates may occur at the remaining vertices without a balance assumption. Corollary~\ref{comp:closure} gives FEI+ for the resulting function with the maximum of the applicable gate constants. Corollary~\ref{comp:full-entropy} then supplies the corresponding full entropy bounds.

\section{Joint selector entropy}\label{sec:selectors}

The communication interpretation of FEI and coding arguments for
composition were developed by Wan, Wright and Wu~\cite{WWW14}.
Bounds through average unambiguous parity certificate complexity
were proved by Arunachalam, Chakraborty, Kouck\'y, Saurabh and
de Wolf~\cite{ACKSW21}. The conditional calculations below give exact
entropy and influence identities for selectors, including affine cells
and overlapping leaf variable sets.

Let $A_1,\ldots,A_m$ partition a uniform control cube. Write
$a_j=\mathbf 1_{A_j}$ and $p_j=\mathbb E a_j>0$.
Let $g_j$ be balanced Boolean functions on a data cube disjoint from the
control coordinates. Assume that their nonempty Fourier supports are
pairwise disjoint. The data variable sets themselves need not be disjoint.
Set
\[
 F(z,y)=\sum_{j=1}^m a_j(z)g_j(y).
\]
This is Boolean because exactly one selector is active.
Define
\[
 q_j(T)=\frac{\widehat a_j(T)^2}{p_j},\qquad
 d_j=\sum_T |T|q_j(T),\qquad D=\sum_jp_jd_j.
\]
Each $q_j$ is a probability law, by Parseval.

\begin{proposition}[Exact selector identities]\label{prop:selector-identities}
Under these assumptions,
\begin{align*}
 \HH(F)&=\HH(p)+\sum_jp_j\bigl(\HH(q_j)+\HH(g_j)\bigr),\\
 \I(F)&=D+\sum_jp_j \I(g_j).
\end{align*}
\end{proposition}
\begin{proof}
For a nonempty data Fourier set $U$ in the support of $\widehat g_j$,
\[
 \widehat F(T\cup U)=\widehat a_j(T)\widehat g_j(U).
\]
The support assumption makes $j$ unique, and balancedness excludes the
empty data set. Consequently the spectral sample identifies its selector
label $J$, with probability $p_j$. Conditional on $J=j$, the control and
data parts are independent with laws $q_j$ and $\widehat g_j^2$.
The two identities follow from the entropy chain rule and additivity of
cardinality across the two coordinate blocks.
\end{proof}

\begin{corollary}[Affine-subspace selectors]\label{cor:affine-selectors}
Suppose every $A_j$ is an affine subspace of codimension $c_j$ over
$\mathbb F_2$. Let $r_j$ be the number of coordinates occurring in its
dual linear subspace. Then
\[
 \HH(F)=2\sum_jp_jc_j+\sum_jp_j\HH(g_j),\qquad
 \I(F)=\frac12\sum_jp_jr_j+\sum_jp_j\I(g_j).
\]
In particular, if
$\HH(g_j)\le C\bigl(\I(g_j)-1\bigr)$ for all $j$, then
\[
 \HH(F)\le\max\{4,C\}\bigl(\I(F)-1\bigr).
\]
For coordinate subcubes, $r_j=c_j$. For distinct nonconstant character
leaves $g_j=\pm\chi_{B_j}$, the entropy is exactly
$2\sum_jp_jc_j$ and the influence is
$\frac12\sum_jp_jr_j+\sum_jp_j|B_j|$.
\end{corollary}
\begin{proof}
The normalized indicator spectrum is uniform on the dual subspace, so
its entropy is $c_j$. Every coordinate present in that subspace is
occupied with probability $1/2$, giving mean degree $r_j/2$.
Also $p_j=2^{-c_j}$, whence $\HH(p)=\sum_jp_jc_j$.
The inequality follows from $c_j\le r_j$, $\sum_jp_j=1$, and the exact
selector identities.
\end{proof}

For ordinary addressing with $k$ controls and $2^k$ distinct data bits,
\[
 \HH(F)=2k,\qquad \I(F)=1+\frac k2,\qquad \HH(F)=4(\I(F)-1).
\]
Thus the constant $4$ in this class is sharp. Its Fourier-set description
uses the selected label once and then describes the whole control subset.

\subsection{A closure statement for arbitrary selector cells}

Put $h_j=2a_j-1$ and
$\HH^+(h_j)=\sum_{T\ne\varnothing}\widehat h_j(T)^2
\log(1/\widehat h_j(T)^2)$. Then
\[
 D=\frac14\sum_j \I(h_j),\qquad
 B:=\sum_jp_j(1-p_j)=\frac14\sum_j\operatorname{Var}(h_j).
\]
Direct substitution gives the exact selector cost
\[
 \HH(p)+\sum_jp_j\HH(q_j)
   =\frac14\sum_j \HH^+(h_j)+R(p),\qquad
 R(p)=2\sum_jp_j(1-p_j)+2\sum_jp_j^2\log\frac1{p_j}.
\]

Two elementary bounds suffice:
\[
 R(p)\le 2\HH(p)\le 4D.
\]
Indeed,
\[
 \sum_jp_j(1-p_j)\log\frac1{p_j}
 =\sum_{i<j}p_ip_j\log\frac1{p_ip_j}
 \ge 2\sum_{i<j}p_ip_j=\sum_jp_j(1-p_j),
\]
using $p_ip_j\le1/4$. This proves the first inequality. The second is
the edge-isoperimetric inequality \cite{OD14} applied to each cell:
$\I(h_j)\ge2p_j\log(1/p_j)$.

\begin{corollary}[Selector closure from full FEI bounds]\label{cor:selector-full}
If $\HH(h_j)\le C_s \I(h_j)$ and
$\HH(g_j)\le C_g(\I(g_j)-1)$ for all $j$, then
\[
 \HH(F)\le\max\{C_s+4,C_g\}(\I(F)-1).
\]
No bound on the number of selectors or the overlap of their control
variable sets is needed.
\end{corollary}
\begin{proof}
The selector entropy is at most $(C_s+4)D$.
The leaf entropy is at most $C_g\sum_jp_j(\I(g_j)-1)$.
Their sum has the stated bound by the exact influence identity.
\end{proof}

The corresponding strengthened inequalities compose without the
additional $4$ in the selector coefficient.

\begin{lemma}[Quadratic probability correction]
Every probability vector satisfies
\[
 \sum_jp_j^2\log\frac1{p_j}\le\sum_jp_j(1-p_j),
 \qquad R(p)\le4B.
\]
\end{lemma}
\begin{proof}
Set $u(t)=t^2(1+\log(1/t))$ with $u(0)=0$.
On $0<t\le1/2$, the ratio $u(t)/t$ is increasing and at most $1$.
If all probabilities are at most $1/2$, then $\sum_j u(p_j)\le1$.
The case of a probability equal to $1$ is immediate. Otherwise, if there is one probability $p>1/2$, put $q=1-p$.
Monotonicity of $u(t)/t$ gives
$\sum_{j:p_j\ne p}u(p_j)\le u(q)$.
It remains to show $u(p)+u(q)\le1$.

Let $r=q/p\in(0,1]$ and
\[
 A(r)=(1+r^2)\ln(1+r)-r^2\ln r.
\]
The binary inequality is equivalent to $A(r)\le2r\ln2$.
Direct differentiation gives
\[
 A''(r)=2\ln\frac{1+r}{r}-\frac2{1+r}-\frac2{(1+r)^2},
 \qquad
 A'''(r)=\frac{2(r-1)}{r(1+r)^3}\le0.
\]
Thus $A''$ decreases from $+\infty$ to $2\ln2-3/2<0$,
and changes sign exactly once. The function
$B_0(r)=rA'(r)-A(r)$ has derivative $rA''(r)$ and satisfies
$B_0(0)=B_0(1)=0$, where the value at zero is by continuity.
It follows that $B_0(r)\ge0$.
Hence $A(r)/r$ is increasing and at most $A(1)=2\ln2$.
This proves $\sum_j u(p_j)\le1$, which is the first inequality.
The bound for $R$ follows from its exact formula.
\end{proof}

\begin{corollary}[Constant-preserving strong selector closure]\label{cor:selector-strong}
If
\[
 \HH^+(h_j)\le C_s\bigl(\I(h_j)-\operatorname{Var}(h_j)\bigr),
 \qquad \HH(g_j)\le C_g\bigl(\I(g_j)-1\bigr)
\]
for every selector and leaf, then
\[
 \HH(F)\le\max\{C_s,C_g,4\}\bigl(\I(F)-1\bigr).
\]
\end{corollary}
\begin{proof}
The exact correction identity bounds the selector entropy by
$C_s(D-B)+4B$. Poincar\'e's inequality for the cells gives $0\le B\le D$,
so this is at most $\max\{C_s,4\}D$.
Add the leaf contribution and use the exact influence identity.
\end{proof}

For example, if each cell function $h_j$ is regular and the leaves are
distinct nonconstant characters, then
\[
 \HH(F)\le(\Creg+4)\bigl(\I(F)-1\bigr),
 \qquad \Creg+4<12.937.
\]
For regular cell functions and balanced regular leaves with pairwise
disjoint nonempty Fourier supports, Theorem~\ref{thm:A} and
Corollary~\ref{cor:selector-strong} give
\[
 \HH(F)\le4\Creg\bigl(\I(F)-1\bigr).
\]

\subsection{Multiplexers with shared data}\label{sec:shared-multiplexers}

Let $x\in\bits^a$ be a uniform selector, independent of a uniform data
vector $y\in\bits^n$, and let $(g_b)_{b\in\bits^a}$ be arbitrary Boolean
functions on that common data cube. Put
\[
 F(x,y)=g_x(y),\qquad N=2^a.
\]
The branches may share data variables and need not be balanced. Write
$\I_{\rm sel}(F)$ and $\I_{\rm data}(F)$ for the influence sums over
selector and data coordinates, respectively.

\begin{proposition}[Exact shared-data identities]\label{prop:shared-data-collision}
For each data Fourier set $T$, put
\[
 R_T=\frac1N\sum_b\wh g_b(T)^2.
\]
For $R_T>0$, define
\[
 r_T(b)=\frac{\wh g_b(T)^2}{NR_T},\qquad
 u_T(b)=\frac{\wh g_b(T)}{\sqrt{NR_T}},\qquad
 \nu_T(S)=\left(\frac1{\sqrt N}\sum_b\chi_S(b)u_T(b)\right)^2.
\]
Then $r_T$ and $\nu_T$ are probability vectors, and
\begin{equation}\label{eq:shared-data-collision}
 \HH(F)=a+\frac1N\sum_b\HH(g_b)
       +\sum_{T:R_T>0}R_T\bigl[\HH(\nu_T)-\HH(r_T)\bigr].
\end{equation}
Here $\HH$ on probability vectors denotes Shannon entropy. In particular,
\begin{equation}\label{eq:shared-data-upper}
 \HH(F)\le2a+\frac1N\sum_b\HH(g_b).
\end{equation}
If $b^{(j)}$ denotes $b$ with selector coordinate $j$ flipped, then
\begin{align}
 \I_{\rm data}(F)&=\frac1N\sum_b\I(g_b),\label{eq:shared-data-influence}\\
 \Inf_{x_j}(F)&=\frac1N\sum_b\P_y[g_b(y)\ne g_{b^{(j)}}(y)],
       \label{eq:shared-selector-influence}\\
 \I_{\rm sel}(F)&=\sum_{T:R_T>0}R_T\sum_S|S|\nu_T(S).
       \label{eq:shared-selector-spectrum}
\end{align}
\end{proposition}
\begin{proof}
The Fourier coefficients are
\[
 \wh F(S,T)=\frac1N\sum_b\chi_S(b)\wh g_b(T).
\]
Parseval in the selector coordinates gives
$\sum_S\wh F(S,T)^2=R_T$. Hence $R$ is the data-set marginal of the
squared Fourier law of $F$, and $\nu_T$ is its conditional selector-set
law. Consider also the probability law
\[
 \P(B=b,T)=\frac1N\wh g_b(T)^2.
\]
Its entropy is $a+N^{-1}\sum_b\HH(g_b)$, its data-set marginal is $R$,
and its conditional label law is $r_T$. The entropy chain rule gives
\[
 \HH(R)=a+\frac1N\sum_b\HH(g_b)-\sum_T R_T\HH(r_T),
 \qquad
 \HH(F)=\HH(R)+\sum_T R_T\HH(\nu_T),
\]
which proves \eqref{eq:shared-data-collision}, including the terms with
$T=\varnothing$. The inequalities $\HH(r_T)\ge0$ and
$\HH(\nu_T)\le a$ give \eqref{eq:shared-data-upper}.
Conditioning on the physical selector proves
\eqref{eq:shared-data-influence} and \eqref{eq:shared-selector-influence}.
Taking the mean selector-set size in the conditional spectral law gives
\eqref{eq:shared-selector-spectrum}.
\end{proof}

\begin{corollary}[Orthogonal neighboring branches]\label{cor:orthogonal-multiplexers}
Suppose that
\[
 \E[g_b g_{b^{(j)}}]=0\quad\text{for every }b,j,
 \qquad \HH(g_b)\le C\I(g_b)\quad\text{for every }b.
\]
Then $\I_{\rm sel}(F)=a/2$ and
\[
 \HH(F)\le\max\{4,C\}\I(F).
\]
More generally, the assumption $\I_{\rm sel}(F)\ge ca$, with $c>0$,
gives the constant $\max\{2/c,C\}$.
\end{corollary}
\begin{proof}
For Boolean branches, orthogonality is equivalent to disagreement
probability $1/2$. Equation~\eqref{eq:shared-selector-influence} therefore
gives $\I_{\rm sel}(F)=a/2$. In \eqref{eq:shared-data-upper}, the selector
term is $4\I_{\rm sel}(F)$ and the average branch entropy is at most
$C\I_{\rm data}(F)$. Adding proves the first assertion. The second uses
$2a\le(2/c)\I_{\rm sel}(F)$ instead.
\end{proof}

\begin{corollary}[Recursive regular multiplexers]\label{cor:recursive-multiplexers}
Start with regular Boolean functions as leaves. At each later node,
select between previously constructed branches using a fresh uniform
selector whose coordinates occur in none of those branches. Allow the
branches to share their data variables, and suppose neighboring branches
at each node are orthogonal under the common uniform data law. Every
function obtained by finitely many such steps satisfies
\[
 \HH(F)\le\Creg\I(F),\qquad \Creg<8.937.
\]
The constant is independent of the number of branches, the depth, and
the pattern of data reuse.
\end{corollary}
\begin{proof}
Theorem~\ref{thm:A} gives the constant $\Creg$ at every regular leaf.
Extending a branch by irrelevant data coordinates changes neither its
entropy nor its influence. Thus the branches at each node can be regarded
as functions on their common data cube. Corollary~\ref{cor:orthogonal-multiplexers}
preserves the constant because $\Creg>4$. Induction proves the claim.
\end{proof}

For a concrete pair with shared data, take
\[
 g_+(y)=\operatorname{Maj}_3(y_1,y_2,y_3),\qquad
 g_-(y)=\operatorname{Maj}_3(y_1,-y_2,y_4).
\]
Their common nonzero Fourier indices are $\{1\}$ and $\{2\}$, whose
inner-product contributions cancel. The two branches are therefore
orthogonal despite sharing two inputs. Their one-bit multiplexer has
two spectral atoms of mass $1/4$ and eight of mass $1/16$, and hence
$\HH(F)=3$ and $\I(F)=2$.

\subsection{Unrestricted selectors and the optimal FEI constant}

Given a balanced Boolean $g$, let
\[
 M_g(x,y_1,y_2)=\frac{1+g(x)}2y_1+\frac{1-g(x)}2y_2.
\]
The exact identities give
\[
 \HH(M_g)=2+\frac12\HH(g),\qquad \I(M_g)=1+\frac12\I(g).
\]
Consequently FEI on the class of these two-cell selectors is equivalent
to full FEI, with the same optimal universal full-FEI constant.
To prove the nontrivial implication without an additive loss, take an
arbitrary Boolean $f$, its tensor power on $t$ independent input blocks,
and balance it with a fresh sign:
\[
 g_t(x^{(1)},\ldots,x^{(t)},z)
   =z\prod_{r=1}^t f(x^{(r)}).
\]
Then
\[
 \HH(g_t)=t\HH(f),\quad \I(g_t)=t\I(f)+1,
\quad \HH(M_{g_t})=2+\tfrac t2\HH(f),\quad
 \I(M_{g_t})=\tfrac32+\tfrac t2\I(f).
\]
If $\HH(M_g)\le C \I(M_g)$ for every balanced $g$, divide the corresponding
inequality by $t/2$ and let $t\to\infty$ to obtain $\HH(f)\le C \I(f)$.
The converse holds because every $M_g$ is Boolean.

The entropy of choosing a selector label is always bounded by $2D$.
For arbitrary cells, the unresolved term is the normalized Fourier
entropy inside the selected cell. A cheap label description alone does
not establish a uniform bound for that term.

\subsection{Rare activation and the optimal FEI constant}\label{sec:rare-equivalence}

Rare-activation reductions to small-influence functions appear in
Wan--Wright--Wu~\cite{WWW14} and Shalev~\cite{Shalev18}.
The following exact calculation records a version that preserves the
optimal universal constant.

Let $g$ be balanced and Boolean, let $A_r$ be the indicator of one point
on $r\ge1$ fresh control bits, and set
\[
 p=2^{-r},\qquad F=1+A_r(g-1).
\]
The function $F$ is Boolean and $\P(F=-1)=p/2$.
\begin{proposition}[Exact activation identities]\label{prop:activation}
With this notation,
\begin{align}
 \HH(F)&=p\HH(g)+2rp(2-p)
       +(1-p)^2\log\frac1{(1-p)^2},\label{eq:activation-entropy}\\
 \I(F)&=p\bigl(\I(g)+r\bigr).\label{eq:activation-influence}
\end{align}
\end{proposition}
\begin{proof}
Write $T$ for a data Fourier set and $U$ for a control Fourier set.
For $T\ne\varnothing$, the coefficient is
$p\varepsilon_U\wh g(T)$, where $\varepsilon_U\in\{-1,1\}$ depends
only on the selected control point. For $T=\varnothing$ and
$U\ne\varnothing$, it is $-p\varepsilon_U$; the constant coefficient is
$1-p$. Summing the nonempty-data terms over the $2^r$ control sets gives
$p\HH(g)+2rp$. The nonempty-control, empty-data terms contribute
$2rp(1-p)$, and the constant term gives the last term in
\eqref{eq:activation-entropy}.

Each data influence is multiplied by $p$. A control bit is pivotal
precisely when all other controls match the specified point and $g=-1$,
an event of probability $2^{-(r-1)}/2=p$. Summing proves
\eqref{eq:activation-influence}.
\end{proof}

Let $C_{\rm FEI}$ be the supremum of $\HH(f)/\I(f)$ over nonconstant
Boolean functions, allowing the value $+\infty$.
For $\varepsilon,\delta>0$, let $C_{\varepsilon,\delta}^{\rm act}$ be
the same supremum restricted to the above balanced-core activations
satisfying $0<\I(F)\le\varepsilon$ and $\P(F=-1)\le\delta$.
\begin{corollary}\label{cor:activation-equivalence}
For every $\varepsilon,\delta>0$,
\[
 C_{\varepsilon,\delta}^{\rm act}=C_{\rm FEI}.
\]
\end{corollary}
\begin{proof}
The upper bound follows from inclusion of classes. For the reverse
bound, fix any nonconstant Boolean $f$ and balance its tensor powers by
\[
 g_t=z\prod_{j=1}^t f(x^{(j)}),
 \qquad \HH(g_t)=t\HH(f),\qquad \I(g_t)=t\I(f)+1.
\]
The blocks and the sign $z$ are independent. Take
$r_t=\lceil2\log(t+1)\rceil$, $p_t=2^{-r_t}$, and activate $g_t$ on a
point of $r_t$ fresh controls. Then $\I(F_t)\to0$ and
$\P(F_t=-1)\to0$. Also
\[
 \frac{(1-p)^2\log(1/(1-p)^2)}p\longrightarrow2\log e
 \qquad(p\downarrow0).
\]
Dividing \eqref{eq:activation-entropy} and
\eqref{eq:activation-influence} by $p_t t$ therefore gives
\[
 \frac{\HH(F_t)}{\I(F_t)}\longrightarrow
 \frac{\HH(f)}{\I(f)},
\]
since $r_t/t\to0$. For all sufficiently large $t$, the prescribed
$\varepsilon,\delta$ bounds hold. Taking the supremum over $f$ proves
the assertion, including the possibility of an infinite supremum.
\end{proof}

Thus a uniform estimate for arbitrary sparsely activated cores would
resolve full FEI. Rarity and small total influence alone impose no
smaller optimal constant. This conclusion concerns arbitrary cores;
the structured gate hypotheses in the preceding theorems supply
additional information.

\section{Further consequences and questions}\label{sec:consequences}

\subsection{Spectral concentration for read-\texorpdfstring{$k$}{k} DNF}

Entropy bounds give sparse Fourier approximations directly. We record the consequence for DNF to make the dependence on the approximation error explicit. This is the concentration problem studied by Mansour \cite{Mansour95}; sharper sparsity bounds as a function of the read number are known from Lecomte and Tan \cite{LT21}.

\begin{corollary}\label{cor:dnf-concentration}
Let $f$ have a read-$k$ DNF representation with $s\ge1$ terms. For every $0<\eps<1$, there is a family $\mathcal A$ of Fourier sets such that
\[
 \sum_{S\in\mathcal A}\wh f(S)^2\ge1-\eps,
 \qquad
 |\mathcal A|\le(2s)^{\,2(41+7\log k)/\eps}.
\]
Thus deleting the coefficients outside $\mathcal A$ gives a polynomial with squared $L^2$ error at most $\eps$.
\end{corollary}
\begin{proof}
The result is immediate if $f$ is constant. Otherwise, after deleting contradictory terms and repeated literals, at most $s$ nonempty terms remain. On an input satisfying the DNF, let $L$ be the least width of a satisfied term, and put $L=0$ on other inputs. Every coordinate on which a positive input is sensitive belongs to each satisfied term, so its sensitivity is at most $L$. Every boundary edge has one positive endpoint. Therefore
\[
 \I(f)\le2\E L,
 \qquad
 \P(L\ge t)\le\min\{1,s2^{-t}\}\quad(t\ge1).
\]
For the latter bound, the event $L\ge t$ requires some term of width at least $t$ to be satisfied. A union bound suffices. Write $s=2^m r$, where $m=\lfloor\log s\rfloor$ and $1\le r<2$. Then
\[
 \E L\le\sum_{t\ge1}\min\{1,s2^{-t}\}
 =m+r\le m+\log r+1=\log s+1.
\]
The scalar inequality $r\le1+\log r$ on $[1,2]$ follows by concavity. Theorem~\ref{thm:readk-intro} with $M=1$ now gives
\[
 \HH(f)\le2(41+7\log k)\log(2s).
\]
If $\HH(f)=0$, one coefficient carries all the mass. Otherwise apply Markov's inequality to $\log(1/\wh f(\mathcal S)^2)$ under the squared Fourier law. The sets with $\wh f(S)^2\ge2^{-\HH(f)/\eps}$ carry mass at least $1-\eps$, and their number is at most $2^{\HH(f)/\eps}$. Parseval gives the approximation assertion.
\end{proof}

\subsection{The product-measure question for weighted thresholds}

The next composition question is whether an absolute constant $C$ satisfies
\[
 \HH_\mu^+(f)\le C\bigl(\I_\mu(f)-\Var_\mu(f)\bigr)
\]
for every weighted linear threshold function and every product measure. Theorem~\ref{bias:strong} proves this for symmetric thresholds. A positive answer for arbitrary weights would remove the balance assumption on weighted internal gates in Corollary~\ref{comp:hybrid}, by the same exact composition identities.

The antichain argument raises a second question. Can a spectral reference measure control the cost of many incomparable coordinate orders without a factor depending on their number? The width bound of Theorem~\ref{thm:width} and the overlap bound of Theorem~\ref{thm:readk-intro} provide two ways to limit that cost. The unrestricted FEI conjecture requires an argument without either restriction.

\section*{Acknowledgements and funding}
The final manuscript was audited by Claude Opus 5.5 (Anthropic), which checked the proofs, recomputed the numerical constants in exact rational and high-precision arithmetic, and reran the finite computations in the repository.

The research presented in this paper was supported by the European Research
Council (ERC) under the European Union's Horizon Europe research and
innovation programme (grant agreement No.~101041711), by the Simons
Foundation as part of the Collaboration on the Mathematical and Scientific
Foundations of Deep Learning, by Heights Labs, by Convex Nexus Capital, by
the Israel Science Foundation (grant number 2258/19), and by the Israel
Science Foundation (ISF Grant 4101/25). Computational resources were
provided in part by Google Cloud through its education credits program
(credit 530550103).

\section*{Declaration of competing interest}
The author holds equity in Heights Labs and in Convex Nexus Capital and
has a management role at Convex Nexus Capital. These entities are included
in the funding declaration above.

\begin{samepage}
\section*{Declaration of generative AI and AI-assisted technologies}
The paper comes out of Math Brain, an ongoing project on the use of AI systems for research mathematics. Earlier versions of the rank, bounded-overlap, and biased-threshold arguments were developed in September 2026 with Claude (Anthropic). In October 2026, OpenAI's GPT-6 Astra Pro, working in ChatGPT with collaborating agents, was used to extend the results, improve the constants, reconstruct and check the proofs, search the literature, carry out finite computations, and prepare the manuscript. Separate agents reviewed the rank, bounded-overlap, biased-threshold, joint-witness, selector, and rare-activation arguments. Claude Opus 5.5 (Anthropic) audited the final manuscript, including the proofs, the numerical constants and the finite computations. These checks are not peer review. The author directed the work and is responsible for the content of the paper and for decisions concerning its dissemination.
\par
\end{samepage}

\section*{Data availability}
\url{https://github.com/yspennstate/fourier-entropy-threshold-functions}

\end{document}